\documentclass[a4paper, UKenglish, cleveref, autoref, thm-restate, numberwithinsect]{lipics-v2021}

\newboolean{arXivLongVersion}
\setboolean{arXivLongVersion}{true} %

\pdfoutput=1 %
\ifthenelse{\boolean{arXivLongVersion}}{%
\hideLIPIcs  %
}{%
\relatedversiondetails{Full Version}{https://arxiv.org/abs/XXX}
}

\usepackage{ifthen}

\usepackage{refcount}

\usepackage{amsfonts}

\usepackage{mathtools}

\usepackage[only, llbracket, rrbracket, fatsemi]{stmaryrd}

\usepackage{amsthm}
\usepackage{thmtools,thm-restate}

\usepackage{tikz}
\usetikzlibrary{arrows, calc, shapes, snakes, automata, fit, positioning, intersections}
\tikzset{>=stealth'} %

\tikzstyle{every picture} = [style=semithick]
\tikzstyle{every node}    = [font=\small]
\tikzstyle{every state}   = [thick, minimum size=1mm, inner sep=2pt]

\tikzstyle{initial}   = [initial   by arrow, initial   text=, initial   distance=4mm]
\tikzstyle{accepting} = [accepting by arrow, accepting text=, accepting distance=4mm]

\usepackage{algorithm}
\usepackage{algpseudocodex}
\algrenewcommand\algorithmicrequire{\textbf{Input:}}
\algrenewcommand\algorithmicensure{\textbf{Output:}}

\usepackage{subcaption}

\usepackage{hyperref}

\theoremstyle{remark}

\newcommand{\setN}{\mathbb{N}}
\newcommand{\setZ}{\mathbb{Z}}
\newcommand{\setQ}{\mathbb{Q}}

\renewcommand{\vec}[1]{{\mathbf #1}}
\newcommand{\norm}[1]{|#1|}

\newcommand{\per}[1]{\operatorname{Per}(#1)}

\newcommand{\lin}[1]{\operatorname{lin}(#1)}
\renewcommand{\dim}[1]{\operatorname{dim}(#1)}
\newcommand{\sdim}[1]{\operatorname{sdim}(#1)}

\newcommand{\poststar}[2]{\operatorname{Post}^*_{#1}(#2)}
\newcommand{\src}[1]{\operatorname{src}(#1)}
\newcommand{\tgt}[1]{\operatorname{tgt}(#1)}

\newcommand{\prestar}[2]{\operatorname{Pre}^*_{#1}(#2)}

\title{A Forward-Only Construction of Semilinear Inductive Invariants for VAS} %

\titlerunning{Forward-Only Inductive Invariants for VAS} %

\author{Clotilde Bizière}{LaBRI, Univ. Bordeaux, CNRS, Bordeaux INP, Talence, France \and University of Warsaw, Poland}{}{https://orcid.org/0009-0003-6469-1170}{Supported by Polish National Science Centre SONATA BIS-12 grant
number 2022/46/E/ST6/00230}
\author{Jérôme Leroux}{LaBRI, Univ. Bordeaux, CNRS, Bordeaux INP, Talence, France}{}{https://orcid.org/0000-0002-7214-9467}{}
\author{Grégoire Sutre}{LaBRI, Univ. Bordeaux, CNRS, Bordeaux INP, Talence, France}{}{https://orcid.org/0009-0004-3839-0005}{}

\authorrunning{C. Bizière, J. Leroux, and G. Sutre}

\Copyright{Clotilde Bizière, Jérôme Leroux, and Grégoire Sutre}



\funding{%
This work was supported by the grant ANR-25-CE48-6933
of the French National Research Agency
(project CoqoPetri).}%

\nolinenumbers %

\EventEditors{Michal Kouck\'{y} and Daniela Petri\c{s}an}
\EventNoEds{2}
\EventLongTitle{51st International Symposium on Mathematical Foundations of Computer Science (MFCS 2026)}
\EventShortTitle{MFCS 2026}
\EventAcronym{MFCS}
\EventYear{2026}
\EventDate{August 24--28, 2026}
\EventLocation{Paris, France}
\EventLogo{}
\SeriesVolume{386}
\ArticleNo{3}

\begin{document}

\maketitle

\begin{abstract}
  The reachability problem for Vector Addition Systems (VAS) is a central decision problem in the theory of infinite-state systems, first solved by Kosaraju and Mayr in the 1980s. An alternative, conceptually simpler approach introduced by Leroux shows that non-reachability is always witnessed by semilinear inductive invariants, yielding a decision procedure by combining an enumeration of runs with a search for such invariants. However, the construction of these invariants relies on a back-and-forth scheme that depends symmetrically on the source and the target. As a result, the invariants are not guaranteed to reflect the structural properties of the VAS, and the construction is difficult to extend to asymmetric models such as Branching VAS.

We introduce a new forward-only construction of semilinear inductive invariants for VAS. Our method builds invariants from the source configuration alone and avoids the need for backward reasoning. This yields invariants that are more canonical and better aligned with the structure of the system. In particular, our method produces periodic inductive invariants for periodic VAS.

Beyond its intrinsic interest, our approach provides a step toward extending invariant-based techniques to Branching VAS.
\end{abstract}

\section{Introduction}\label{sec:introduction}
Vector Addition Systems (VAS), equivalently known as Petri nets, are a fundamental model of computation used to represent concurrent processes arising in areas such as distributed systems, chemical reaction networks, and biological systems. 
A VAS consists of an initial configuration in $\setN^d$ together with a finite set of actions in $\setZ^d$.
The reachability problem asks whether a given target vector can be obtained by iteratively adding actions to the initial configuration while maintaining non-negativity.

The reachability problem for VAS is one of the most fundamental decision problems in the theory of infinite-state systems. It was first posed in the late 1960s~\cite{KarpM69} and its decidability, established in the early 1980s through the works of Mayr and Kosaraju~\cite{Kosaraju,Mayr84}, is widely regarded as a major breakthrough in theoretical computer science.
Despite this early result, the computational complexity of the problem remained open for several decades, until it was finally shown to be Ackermann-complete, more than forty years after its introduction~\cite{DBLP:conf/lics/LerouxS19,DBLP:conf/focs/Leroux21,CO22}.

Two different approaches to the decidability of VAS reachability are known.
The first one, often referred to as the KLM decomposition (after Kosaraju, Lambert, Mayr), is based on a structural analysis of runs. It proceeds by decomposing the system into simpler components until a certain sufficient condition for reachability is met, or the decomposition process is exhausted. While powerful, this approach is technically involved, and the underlying condition long appeared somewhat ad hoc, until it was later given a conceptual explanation in terms of ideals of well-quasi-orders~\cite{LS15}. Moreover, refinements of this approach have led to the Ackermann complexity upper bound~\cite{DBLP:conf/lics/LerouxS19}.

A second, more recent approach was introduced by Leroux in the 2010s~\cite{DBLP:conf/popl/Leroux11,Turing-100:Vector_Addition_Systems_Reachability}. It is based on the use of inductive invariants to characterize non-reachability. More precisely, Leroux showed that if a target configuration is not reachable, then there exists a semilinear set containing the initial configuration, closed under the actions of the system, and excluding the target. This yields a decision procedure by combining two semi-algorithms: one enumerating runs, and the other enumerating semilinear sets in search of such inductive invariants.
\medskip

Numerous extensions of Vector Addition Systems have been introduced over the years, either to model richer computational phenomena or to explore the boundaries of decidability. For several of these extensions, the reachability problem remains decidable—for instance in the presence of nested zero tests~\cite{REINHARDT2008239,VASSnzCGL}. Recently, reachability has also been shown to be decidable for pushdown VAS~\cite{GKM25}, a long-standing open problem. However, it is still open for well-studied models, including unordered data nets~\cite{datanets} and, in particular, Branching Vector Addition Systems (BVAS)~\cite{GGS04,DBLP:journals/dmtcs/VermaG05}.
BVAS extend VAS with branching transitions, allowing one to combine two reachable configurations into a new one. As a result, executions are no longer sequences but trees. 

Extending existing techniques for VAS to BVAS is therefore a natural research direction. However, the classical KLM approach appears difficult to generalize to this setting. In contrast, the inductive invariant approach of Leroux seems more promising: recent results show that reachability sets of BVAS enjoy geometric properties similar to those used in Leroux’s proof for VAS~\cite{Fossacs26}.

A major obstacle, however, lies in the fundamentally asymmetric nature of BVAS. Indeed, Leroux’s construction crucially relies on a symmetry property of VAS: by taking the opposite of each action of a VAS, one obtains a reversed VAS in which every run from $\vec s$ to $\vec t$ in the original VAS can be reversed into a run from $\vec t$ to $\vec s$.
This symmetry breaks down in BVAS, where computations are inherently tree-shaped: the initial configuration is placed at the leaves, while the target appears at the root.
\medskip

Beyond these considerations, Leroux’s construction already presents certain limitations even in the setting of plain VAS. His approach follows a back-and-forth scheme: it alternately expands a semilinear over-approximation $\vec S$ of the forward reachability set from the source and a semilinear over-approximation $\vec T$ of the backward reachability set from the (unreachable) target, while ensuring that there is no run from $\vec S$ to $\vec T$, until the two sets become complementary (which implies that $\vec S$ is inductive). As a result, the inductive invariant obtained depends equally on the source and the target.

However, in many applications, the source and the target play fundamentally different roles. The source is part of the system specification, fixed once and for all, whereas the target represents a query that may vary. This asymmetry is also reflected in other VAS verification problems, such as coverability or boundedness, which are inherently one-sided.

From this perspective, the back-and-forth nature of Leroux’s construction makes the resulting invariants less canonical and harder to relate directly to the structure of the VAS.
A good example is the case of periodic VAS, i.e., VAS whose reachability set is periodic, meaning that it contains the zero vector and is closed under addition. Such systems naturally arise as abstractions of population protocols~\cite{DBLP:conf/podc/AngluinADFP04} and leaderless communication protocols~\cite{10.1145/146637.146681}.
The inductive invariants produced by Leroux’s method for periodic VAS are not always periodic.

\medskip

\noindent \textbf{Our contributions} 
\begin{itemize}
\item[(i)] We propose a new forward-only construction of inductive invariants for VAS.
In contrast with Leroux’s back-and-forth approach, our construction depends primarily on the reachable configurations from the source rather than on the co-reachable configurations from the target and builds invariants in a purely forward manner.\footnote{Our construction still depends on the target but this is unavoidable. Indeed,  since the forward reachability set from the source is not semilinear in general, there is no semilinear inductive invariant disjoint from every unreachable target.} As a result, it yields inductive invariants that more directly reflect the structure of the system. 
\item[(ii)] We introduce the class of periodic VAS defined as the class of VAS with periodic reachability sets. We provide an effective characterization of those VAS proving that the class of periodic VAS is recursive. We prove that the reachability problem for plain VAS can be reduced to the reachability problem for periodic VAS and we explain how semilinear inductive invariant for plain VAS can be computed from semilinear periodic inductive invariant for periodic VAS.
\item[(iii)] We apply our forward-only construction of inductive invariants to the class of periodic VAS and we prove that if the reachability set of a periodic VAS has an empty intersection with a semilinear set of target configurations, then there exists a semilinear periodic inductive invariant that is disjoint from the semilinear target set.
\end{itemize}

\section{Example} \label{sec:example}
These examples illustrate the differences between our forward-only construction and Leroux’s back-and-forth construction.
\ifthenelse{\boolean{arXivLongVersion}}{%
A detailed presentation is given in Appendix~\ref{app:example}, where all necessary definitions are recalled and both constructions are described intuitively.
}{}

\begin{example} \label{ex:1}
    Consider the VAS with actions $\vec A \coloneqq \{(0,1,0), (0,-1,1), (1,2,-2)\}$ and initial configuration $\vec C_0 = \{\vec 0\}$. We take unreachable targets of the form $\vec C_{\text{bad}} \coloneqq \{(x_{\text{bad}},1,0)\}$ for some $x_{\text{bad}} \ge 1$. The back-and-forth construction yields inductive invariants of the form $\setN^3 \setminus ([1,\infty) \times \{(0,0)\} \cup \{(x_{\text{bad}},1,0), (x_{\text{bad}},0,0)\})$, while our forward-only construction returns $\{(0,0,0), (0,1,0)\} \cup ((0,2,0) + \setN^3) \cup ((0,0,1) + \setN^3)$. 
    In particular, the invariant produced by the back-and-forth construction depends strongly on $\vec C_{\text{bad}}$, unlike the forward-only one.
\end{example}

\begin{example} \label{ex:BnF} 
    Consider the periodic VAS with actions $\vec A \coloneqq \{5,6\}$ and initial configuration $\vec C_0 \coloneqq \{0\}$. Let $\vec C_{\text{bad}} = \{14\}$. The back-and-forth construction yields the invariant $\setN \setminus \{1,2,3,4,8,9,14\}$, which is not periodic (it contains $7$ but not $7+7 = 14$). By contrast, the forward-only construction returns the periodic invariant $\setN \setminus \{1,2,3,4,7,8,9,13,14\}$.
\end{example}

\section{Preliminaries}\label{sec:preliminaries}
This section recalls the definitions of periodic sets and semilinear sets as well as classical results about vector addition systems. In the sequel, the set of natural numbers, integers, non-negative rational numbers, and rational numbers are denoted respectively by $\setN$, $\setZ$, $\setQ_{\geq 0}$, and $\setQ$. Vectors as well as sets of vectors are denoted in bold face and operations are extended component-wise. We denote by $\vec{e}_i$ the \emph{$i$th unit vector} of $\setN^d$ defined by $\vec{e}_i(i)=1$ and $\vec{e}_i(j)=0$ for every $j\not=i$. The sum $\vec{X}+\vec{Y}$ of two sets $\vec{X},\vec{Y}\subseteq\setQ^d$ is the set $\{\vec{x}+\vec{y} \mid (\vec{x},\vec{y})\in\vec{X}\times\vec{Y}\}$. We also denote by $\vec{X}+\vec{y}$ and $\vec{x}+\vec{Y}$ the sets $\vec{X}+\{\vec{y}\}$ and $\{\vec{x}\}+\vec{Y}$ where $\vec{x},\vec{y}\in\setQ^d$ and $\vec{X},\vec{Y}\subseteq \setQ^d$. We denote by $\setQ\vec{X}$ and $\setQ_{\geq 0}\vec{X}$ the sets $\{\lambda\vec{x}\mid (\lambda,\vec{x})\in\setQ\times\vec{X}\}$ and $\{\lambda\vec{x}\mid (\lambda,\vec{x})\in\setQ_{\geq 0}\times\vec{X}\}$ respectively.

\subsection{Vector Spaces and Space-dimension}
A \emph{vector space} is a set $\vec{V}\subseteq \setQ^d$ that contains the zero vector, such that $\vec{V}+\vec{V}\subseteq\vec{V}$ and such that $\setQ\vec{V}\subseteq \vec{V}$. The vector space \emph{spanned} by a set $\vec{X}\subseteq \setQ^d$ is the set of finite sums $\lambda_1\vec{x}_1+\cdots+\lambda_k\vec{x}_k$ of $k\in\setN$ vectors $\vec{x}_j\in \vec{X}$ scaled by rational numbers $\lambda_j\in \setQ$. Recall that every vector space $\vec{V}\subseteq \setQ^d$ is spanned by a finite set of vectors. The minimal number of vectors spanning $\vec{V}$ is called the \emph{dimension} of $\vec{V}$ and it is denoted by $\dim{\vec{V}}$. Let us recall that $\dim{\vec{V}}=0$ if and only if $\vec{V}=\{\vec{0}\}$. Moreover $\dim{\vec{V}}<\dim{\vec{W}}$ for every vector spaces $\vec{V}\subset\vec{W}$, and in particular $\dim{\vec{V}}\leq d$ for every vector space $\vec{V}\subseteq \setQ^d$. Finally, if $\vec{V}$ is spanned by a set $\vec{X}\subseteq \setQ^d$ then $\vec{V}$ is spanned by a subset of $\dim{\vec{V}}$ vectors in $\vec{X}$.

\medskip

The \emph{space-dimension}~\cite[Section~5]{DBLP:conf/popl/Leroux11} of a set $\vec{X}\subseteq \setQ^d$ is the minimal $r\in\{-1,\ldots,d\}$ such that there exists a sequence $\vec{x}_1,\ldots,\vec{x}_k\in\setQ^d$ and a sequence $\vec{V}_1,\ldots,\vec{V}_k$ of vector spaces $\vec{V}_j$ such that $\dim{\vec{V}_j}\leq r$ and such that $\vec{X}\subseteq \vec{x}_1+\vec{V}_1\cup\cdots\cup \vec{x}_k+\vec{V}_k$. We denote by $\sdim{\vec{X}}$ the \emph{space-dimension} of $\vec{X}$. Observe that $\sdim{\vec{X}}=-1$ iff $\vec{X}$ is empty and $\sdim{\vec{X}}=0$ iff $\vec{X}$ is a non-empty finite set. Note also that $\sdim{\vec{X}\cup\vec{Y}}=\max\{\sdim{\vec{X}},\sdim{\vec{Y}}\}$,  $\sdim{\vec{X}}\leq \sdim{\vec{Y}}$ if $\vec{X}\subseteq\vec{Y}$, and $\sdim{\vec{x}+\vec{X}}=\sdim{\vec{X}}$ for every $\vec{x}\in\setQ^d$ and $\vec{X},\vec{Y}\subseteq \setQ^d$.

\begin{example}\label{ex:D}
 The space dimension of $(\{0\}\times\setN)\cup(\setN\times\{0\})$ is $1$. The space dimension of $\{(n,2^{n+1}) \mid n\in\setN\}$ is $2$. 
\end{example}

The following lemma and corollary shows that the space-dimension of a vector space coincides with its dimension. The proofs are similar to~\cite[Lemma 7.3]{Turing-100:Vector_Addition_Systems_Reachability} and~\cite[Lemma 7.4]{Turing-100:Vector_Addition_Systems_Reachability}.
\ifthenelse{\boolean{arXivLongVersion}}{%
They are also given in appendix for sake of completeness.
}{}
\begin{restatable}{lemma}{leminsecable}\label{lem:insecable}
  Let $\vec{P}\subseteq \setQ^d$ such that $\vec{P}+\vec{P}\subseteq \vec{P}$. Let $k\in\setN_{>0}$, $\vec{V}_1,\ldots,\vec{V}_k$ be a sequence of vector spaces of $\setQ^d$ and $\vec{x}_1,\ldots,\vec{x}_k$ be a sequence of vectors in $\setQ^d$. We have $\vec{P}\subseteq \bigcup_{j=1}^k\vec{x}_j+\vec{V}_j$ if, and only if, there exists $j\in\{1,\ldots,k\}$ such that $\vec{P}\subseteq \vec{x}_j+\vec{V}_j$.
\end{restatable}
\begin{proof}[Proof Sketch.]
  We observe that if there exists $\vec{p}_0\in \vec{P}\setminus (\vec{x}_k+\vec{V}_k)$ then $\vec{p}_0+\vec{P}\subseteq \bigcup_{j=1}^{k-1}\vec{x}_j+\vec{V}_j$. This property is obtained by observing that if $\vec{p}_0+\setN\vec{p}$ where $\vec{p}\in\vec{P}$ has an infinite intersection with some $\vec{x}_j+\vec{V}_j$, then $\vec{p}_0+\setQ\vec{p}$ is included in $\vec{x}_j+\vec{V}_j$ since $\vec{V}_j$ is a vector space. The proof of the lemma is then obtained by induction on $k$.
\end{proof}

\begin{restatable}{corollary}{cordimspacedim}\label{cor:dimspacedim}
  We have $\sdim{\vec{P}}=\dim{\vec{V}}$ where $\vec{V}$ is the vector space spanned by a non-empty set $\vec{P}\subseteq \setQ^d$ such that $\vec{P}+\vec{P}\subseteq\vec{P}$.
\end{restatable}

\subsection{Periodic Sets and Semilinear Sets}
A set $\vec{P}\subseteq \setN^d$ is said to be \emph{periodic} if it is a submonoid of $(\setN^d,+)$, i.e., if it is such that $\vec{0}\in\vec{P}$ and $\vec{P}+\vec{P}\subseteq \vec{P}$. The periodic set \emph{spanned} by a set $\vec{G}\subseteq\setN^d$ is the set $\per{\vec{G}}$ of finite sums $\vec{g}_1+\cdots+\vec{g}_k$ of $k\in\setN$ vectors $\vec{g}_j\in \vec{G}$. Notice that $\per{\vec{G}}$ is the minimal periodic set that contains $\vec{G}$. A \emph{periodic set} $\vec{P}\subseteq \setN^d$ is said to be \emph{finitely-generated} if $\vec{P}=\per{\vec{G}}$ for some finite set $\vec{G}\subseteq \setN^d$.

A \emph{linear set} is a set $\vec{L}\subseteq \setN^d$ of the form $\vec{b}+\per{\vec{G}}$ where $\vec{b}\in\setN^d$ is called the \emph{base} and $\vec{G}\subseteq\setN^d$ is a finite set of vectors called the \emph{periods}. A pair $(\vec{b},\vec{G})$ is called a \emph{linear-presentation} of $\vec{L}$.  A \emph{semilinear set} $\vec{S}\subseteq \setN^d$ is a finite union $\vec{L}_1\cup\ldots\cup\vec{L}_k$ of linear sets $\vec{L}_j\subseteq \setN^d$ where $k\in\setN$. A finite set $\{(\vec{b}_1,\vec{G}_1),\ldots,(\vec{b}_k,\vec{G}_k)\}$ where $(\vec{b}_j,\vec{G}_j)$ is a linear-presentations of $\vec{L}_j$ is called a \emph{semilinear-presentation} of $\vec{S}$. In the sequel semilinear sets (resp. linear sets) are always assumed to be implicitly given by presentations.

Let us recall that the class of semilinear sets is stable by many operations, and in particular by sum, union, intersection, complement, projection, and Cartesian product since a set is semilinear iff it is definable in the Presburger arithmetic~\cite{gs66}. Moreover, the periodic sets spanned by semilinear sets are semilinear since $\per{\vec{X}\cup\vec{Y}}=\per{\vec{X}}+\per{\vec{Y}}$ for every sets $\vec{X},\vec{Y}\subseteq\setN^d$.

\begin{example}
  A finitely-generated periodic set is clearly semilinear. The converse is not true since the periodic set $\vec{P}=\{(0,0)\}\cup \setN_{>0}^2$ is semilinear but not finitely-generated.
\end{example}

In the sequel, we are going to under-approximate linear sets $\vec{L}$ by linear sets of the form $\vec{L}'\coloneqq\vec{L}+\vec{q}$ for some $\vec{q}\in\setN^d$ such that $\vec{L}'\subseteq \vec{L}$. The following lemma shows that the space-dimension of the difference $\vec{L}\setminus \vec{L}'$ is strictly smaller than the space-dimension of $\vec{L}$.
\begin{restatable}{lemma}{lemspacedifflinear}\label{lem:spacediff:linear}
  For every linear set $\vec{L} \subseteq \setN^d$ and for every $\vec{q} \in \setN^d$ such that $\vec{L}'\coloneqq(\vec{L} + \vec{q}) \subseteq \vec{L}$, we have:
  $$\sdim{\vec{L} \setminus \vec{L}'} < \sdim{\vec{L}}$$
\end{restatable}

\begin{remark}
  If $\vec{L}=\vec{b}+\vec{P}$ is a linear set, a vector $\vec{q}$ is such that $\vec{L}+\vec{q}\subseteq \vec{L}$ iff $\vec{q}\in\vec{P}$.
\end{remark}

\subsection{Vector Addition Systems (VAS)}
\label{subsec:vas}
A \emph{$d$-dim vector addition system} (\emph{$d$-VAS} or just \emph{VAS} for short) is a finite set $\vec{A}$ of vectors in $\setZ^d$ called \emph{actions}. A \emph{configuration} is a vector in $\setN^d$. An \emph{initialized VAS}, or just a \emph{VAS} if there is no ambiguity, is a pair $(\vec{A},\vec{C}_0)$ where $\vec{C}_0\subseteq \setN^d$ is a semilinear set of \emph{initial configurations}. The \emph{reachability set} of a VAS $(\vec{A},\vec{C}_0)$ is the set of configurations $\poststar{\vec{A}}{\vec{C}_0}$ defined as the unique minimal (for the inclusion) set of configurations satisfying the following monovariate system of constraints:
\begin{align}
  \label{eq:inductive:init} & \vec{X}\supseteq \vec{C}_0\\ 
  & \vec{X}\supseteq (\vec{X}+\vec{a})\cap\setN^d & \forall \vec{a}\in\vec{A}\label{eq:inductive:step}
\end{align}

An \emph{inductive invariant} for an initialized VAS $(\vec{A},\vec{C}_0)$, resp. for a VAS $\vec{A}$, is a set of configurations $\vec{X}\subseteq \setN^d$ satisfying~(\ref{eq:inductive:init}) and~(\ref{eq:inductive:step}), resp. only~(\ref{eq:inductive:step}). Clearly, the reachability set of $(\vec{A},\vec{C}_0)$ is the minimal for the inclusion inductive invariant for $(\vec{A},\vec{C}_0)$. Since semilinear sets are effective for many operations, it follows the class of semilinear inductive invariants (for a initialized VAS or a VAS) is recursive.

\medskip

The reachability set can be equivalently expressed by introducing the classical notion of runs as follows. We associate with an action $\vec{a}\in\vec{A}$ the binary relation $\xrightarrow{\vec{a}}$ on $\setN^d$ defined by $\vec{x}\xrightarrow{\vec{a}}\vec{y}$ if $\vec{y}=\vec{x}+\vec{a}$. A \emph{run} $\rho$ of a vector addition system $\vec{A}$ is a non-empty word of configurations $\rho=\vec{c}_0\ldots\vec{c}_k$ such that for every $j\in\{1,\ldots,k\}$ there exists an action $\vec{a}_j\in\vec{A}$ such that $\vec{c}_{j-1}\xrightarrow{\vec{a}_j}\vec{c}_{j}$. In that case, we say that $\rho$ is a run from $\vec{c}_0$ to $\vec{c}_k$ labeled by $\sigma=\vec{a}_1\ldots\vec{a}_k$ and we write $\vec{c}_0\xrightarrow{\sigma}\vec{c}_k$. The configurations $\vec{c}_0$ and $\vec{c}_k$ are respectively called the \emph{source} and \emph{target} of $\rho$, and they are denoted by $\src{\rho}$ and $\tgt{\rho}$. The \emph{reachability relation} of a VAS $\vec{A}$ is the binary relation $\xrightarrow{\vec{A}^*}$ on the configurations defined by $\vec{x}\xrightarrow{\vec{A}^*}\vec{y}$ if there exists a run from $\vec{x}$ to $\vec{y}$ labeled by a word in $\vec{A}^*$. Observe that $\poststar{\vec{A}}{\vec{C}_0}$ is the set of configurations $\vec{y}\in\setN^d$ such that there exists a run from a configurations $\vec{x}\in\vec{C}_0$ labeled by a word in $\vec{A}^*$ to $\vec{y}$.

\medskip

The reachability problem for VAS takes as input a VAS $(\vec{A},\vec{C}_0)$ and a semilinear set $\vec{C}_{bad}$ and checks if $\vec{C}_{bad}$ has an empty intersection with the reachability set of $(\vec{A},\vec{C}_0)$. Based on a simple algorithm that enumerates semilinear inductive invariants and runs, the following theorem shows that the reachability problem for VAS is decidable.  
\begin{theorem}[\cite{DBLP:conf/lics/Leroux09,DBLP:conf/popl/Leroux11,Turing-100:Vector_Addition_Systems_Reachability}]\label{thm:inv}
  The reachability set of a VAS $(\vec{A},\vec{C}_0)$ has an empty intersection with a semilinear set $\vec{C}_{bad}$ if and only if there exists a semilinear inductive invariant $\vec{I}$ of $(\vec{A},\vec{C}_0)$ such that $\vec{I}\cap \vec{C}_{bad}$ is empty.
\end{theorem}

\section{Periodic VAS}\label{sec:periovas}
A VAS $(\vec{A},\vec{C}_0)$ is said to be \emph{periodic} if its reachability set is periodic. The following lemma shows that if $\vec{C}_0$ is periodic then $(\vec{A},\vec{C}_0)$ is periodic. It also shows that if a VAS $(\vec{A},\vec{C}_0)$ is periodic, then we can replace $\vec{C}_0$ by the semilinear periodic set $\per{\vec{C}_0}$ without modifying the reachability set. Last but not least, since the inclusion of a semilinear set in the reachability set of a VAS is decidable~\cite{DBLP:conf/lics/Leroux13}, this lemma also proves that the class of periodic VAS is recursive.
\begin{lemma}
  \label{lem:periodic-vas}
  Let $(\vec{A},\vec{C}_0)$ be a VAS. The following properties are equivalent:
  \begin{itemize}
  \item[(i)] $(\vec{A},\vec{C}_0)$ is periodic.
  \item[(ii)] $\poststar{\vec{A}}{\vec{C}_0}=\poststar{\vec{A}}{\per{\vec{C}_0}}$.
  \item[(iii)] $\per{\vec{C}_0}\subseteq \poststar{\vec{A}}{\vec{C}_0}$.
  \end{itemize}
\end{lemma}
\begin{proof}
  $(iii)\Rightarrow (ii)$ is trivial since $\per{\vec{C}_0}\subseteq \poststar{\vec{A}}{\vec{C}_0}$ implies $\poststar{\vec{A}}{\per{\vec{C}_0}}\subseteq \poststar{\vec{A}}{\vec{C}_0}$. For $(i)\Rightarrow (iii)$, just notice that $\vec{C}_0\subseteq \poststar{\vec{A}}{\vec{C}_0}$ implies $\per{\vec{C}_0}\subseteq \poststar{\vec{A}}{\vec{C}_0}$ since $\poststar{\vec{A}}{\vec{C}_0}$ is periodic.

  Finally let us prove $(ii)\Rightarrow(i)$. Let $\vec{P}$ be the reachability set of $(\vec{A},\per{\vec{C}_0})$ and let us prove that $\vec{P}$ is periodic. We clearly have $\vec{0}\in\vec{P}$. Let $\vec{p}_1,\vec{p}_2\in \vec{P}$. There exist a run $\rho_1$ from a configuration $\vec{x}_1\in\per{\vec{C}_0}$ to $\vec{p}_1$ and a run $\rho_2$ from a configuration $\vec{x}_2\in\per{\vec{C}_0}$ to $\vec{p}_2$. Notice that the word $\rho_1+\vec{x}_2$ obtained from $\rho_1$ by adding $\vec{x}_2$ on each configuration is a run from $\vec{x}_1+\vec{x}_2$ to $\vec{p}_1+\vec{x}_2$. Symmetrically, $\rho_2+\vec{p}_1$ is a run from $\vec{p}_1+\vec{x}_2$ to $\vec{p}_1+\vec{p}_2$. We deduce that $\vec{p}_1+\vec{p}_2\in \vec{P}$. Therefore $\vec{P}$ is periodic.
\end{proof}

The following lemma shows that the previous lemma cannot be extended to the projection of VAS reachability sets. 
\begin{restatable}{lemma}{lemundecper}
  We cannot decide if a VAS $(\vec{A},\vec{C}_0)$ with $\vec{A}\subseteq\setN^d\times\setN^e$ for some $d,e\in\setN$ is such that the set $\{\vec{x}\in\setN^d \mid \exists \vec{y}\in\setN^e, (\vec{x},\vec{y})\in\poststar{\vec{A}}{\vec{C}_0}\}$ is periodic.
\end{restatable}

In the rest of this section, we show that periodic VAS naturally occurs when dealing with diagonal semilinear relations~(see \cref{subsec:diagonal-relations}). We then explain how the reachability problem for plain VAS can be reduced to the reachability problem for periodic VAS (see Section~\ref{sub:plain2per}). Finally, in Section~\ref{sub:mainthm} we introduce the main result of this paper.

\subsection{Diagonal Relations}
\label{subsec:diagonal-relations}
This is a folklore result that the reachability relation of a $d$-dim VAS $\vec{A}\subseteq \setN^d$ corresponds to the reachability set of $2d$-dim initialized periodic VAS $(\vec{A}',\{(\vec{0},\vec{0})\})$ defined as follows where $\vec{e}_i$ is the $i$th unit vector (see preliminaries):
$$\vec{A}' = (\{\vec{0}\}\times\vec{A})\cup \{(\vec{e}_i,\vec{e}_i)\mid 1\leq i\leq d\}$$

\medskip

This observation can be extended to the reflexive and transitive closure of any semilinear \emph{diagonal relation} defined as follows. A binary relation $\vec{R}$ over $\setN^d$ is said to be \emph{diagonal} if $(\vec{x},\vec{y})+(\vec{c},\vec{c})\in \vec{R}$ for every $(\vec{x},\vec{y})\in \vec{R}$ and for every $\vec{c}\in\setN^d$. Diagonal relations are clearly stable by union, intersection, composition, and in particular by transitive and reflexive closure. Let us denote by $\vec{R}^*$ the reflexive and transitive closure of a binary relation $\vec{R}$.
\begin{lemma}\label{lem:diag2per}
  The reflexive and transitive closure of a diagonal relation is periodic.  
\end{lemma}
\begin{proof}
  Let $\vec{R}$ be a diagonal reflexive transitive binary relation. Clearly $(\vec{0},\vec{0})\in \vec{R}$ since $\vec R$ is reflexive. Let $(\vec{x}_1,\vec{y}_1)$ and $(\vec{x}_2,\vec{y}_2)$ be pairs in $\vec R$. Since $\vec R$ is diagonal from the membership in $\vec R$ of these two pairs, it follows that $(\vec{x}_1,\vec{y}_1)+(\vec{x}_2,\vec{x}_2)$ and $(\vec{x}_2,\vec{y}_2)+(\vec{y}_1,\vec{y}_1)$ are in $\vec R$. Since $\vec R$ is transitive, we deduce that $(\vec{x}_1+\vec{x}_2,\vec{y}_1+\vec{y}_2)$ is in $\vec R$. Since this pair is equal to $(\vec{x}_1,\vec{y}_1)+(\vec{x}_2,\vec{y}_2)$. We have proved that $\vec R+\vec R\subseteq \vec R$. It follows that $\vec R$ is periodic. 
\end{proof}

The following lemma shows that reflexive and transitive closures of semilinear diagonal relations are related to the model of VAS with states, or equivalently to regular expressions of VAS actions, an extension of the model of VAS with the same expressive power~\cite{KarpM69}. We do not recall the classical definition of regular expressions over a finite alphabet, but we just emphasis that the relation $\xrightarrow{E}$ in the following lemma where $E$ is a regular expression over a $2d$-dim VAS (that plays the role of a finite alphabet) is the binary relation over the pairs in $\setN^d\times\setN^d$ defined by $(\vec{x},\tilde{\vec{x}})\xrightarrow{E}(\vec{y},\tilde{\vec{y}})$ if there exists a word $\sigma$ accepted by $E$ such that $(\vec{x},\tilde{\vec{x}})\xrightarrow{\sigma}(\vec{y},\tilde{\vec{y}})$.
\begin{lemma}
  For every semilinear diagonal relation $R$ over $\setN^d$, we can effectively build a regular expression $E$ over the actions of a $2d$-dim VAS such that for all pairs $(\vec{x},\tilde{\vec{x}})$ and $(\vec{y},\tilde{\vec{y}})$ in $\setN^d\times\setN^d$, we have:
  $$(\vec{x}+\tilde{\vec{x}}) \vec{R}^* (\vec{y}+\tilde{\vec{y}})~~\Longleftrightarrow~~(\vec{x},\tilde{\vec{x}}) \xrightarrow{E} (\vec{y},\tilde{\vec{y}})$$
\end{lemma}
\begin{proof}
  We introduce the finite set $A_{swap}$ of actions of $\setZ^d\times\setZ^d$ defined as the set of vectors $(\vec{e}_i,-\vec{e}_i)$ and $(-\vec{e}_i,\vec{e}_i)$ where $i$ ranges over $\{1,\ldots,d\}$. Notice that for all pairs $(\vec{x},\tilde{\vec{x}})$ and $(\vec{y},\tilde{\vec{y}})$ in $\setN^d\times\setN^d$, we have:
  $$(\vec{x},\tilde{\vec{x}})\xrightarrow{A_{swap}^*}(\vec{y},\tilde{\vec{y}})~~\Longleftrightarrow~~\vec{x}+\tilde{\vec{x}}=\vec{y}+\tilde{\vec{y}}$$
  Now, assume that $\vec R$ is a semilinear diagonal relation given by a semilinear-presentation $\{(b_1,G_1),\ldots,(b_k,G_k)\}$. By replacing $G_j$ by $G_j\cup\bigcup_i\{(\vec{e}_i,\vec{e}_i)\}$, we do not change the semilinear relation $\vec R$ denoted by the semilinear-presentation since $\vec R$ is diagonal. So, w.l.o.g., we can assume that $(\vec{e}_i,\vec{e}_i)\in G_j$ for every $i,j$. Let us introduce the function $\delta:\setZ^d\times\setZ^d\rightarrow\setZ^d\times\setZ^d$ defined by $\delta(\vec{x},\vec{y})=(-\vec{x},\vec{y})$ and observe that for all pairs $(\vec{x},\tilde{\vec{x}})$ and $(\vec{y},\tilde{\vec{y}})$ in $\setN^d\times\setN^d$ and for every $j\in\{1,\ldots,k\}$, we have:
  $$(\vec{x}+\tilde{\vec{x}},\vec{y}+\tilde{\vec{y}})\in b_j+\per{G_j}~~\Longleftrightarrow~~
  (\vec{x},\tilde{\vec{x}})\xrightarrow{A_{swap}^*\delta(b_j)\delta(G_j)^*A_{swap}^*}(\vec{y},\tilde{\vec{y}})$$
  It follows that the following regular expression $E$ satisfies the lemma.
  $$E=A_{swap}^*\cup (\bigcup_{j=1}^k A_{swap}^*\delta(b_j)\delta(G_j)^*A_{swap}^*)^*$$
\end{proof}

\subsection{From Plain VAS to Periodic VAS}\label{sub:plain2per}
Periodic VAS is a central model for population protocols~\cite{DBLP:conf/podc/AngluinADFP04} as well as for communication protocols without leaders~\cite{10.1145/146637.146681} since for those models, the set $\vec{C}_0$ of initial configurations is a periodic semilinear set of the form $\setN(1,0,\ldots,0)$ where the first counter intuitively counts the number of agents in a given initial state. Even if periodic VAS seems to be a strict subclass of VAS, the reachability problem for plain VAS can be easily reduced to the reachability of a periodic VAS by adding just one extra counter that is unchanged by the VAS actions. Let us explain the construction.

Assume that $\mathcal{V}=(\vec{A},\vec{C}_0)$ is a VAS of dimension $d$, i.e. such that $\vec{A}\subseteq \setN^d$ is a finite set of actions and $\vec{C}_0\subseteq\setN^d$ is a semilinear set of initial configurations given by a semilinear-presentation. We introduce the periodic VAS $\mathcal{V}'=(\vec{A}',\vec{C}_0')$ where $\vec{A}'=\vec{A}\times\{0\}$ and $\vec{C}_0'$ is the semilinear periodic set $\per{\vec{C}_0\times\{1\}}$. Now, just observe that the following equality holds.
\begin{align}
  &\poststar{\vec{A}}{\vec{C}_0}\times\{1\}=\poststar{\vec{A}'}{\vec{C}_0'}\cap (\setN^d\times \{1\})
\end{align}
In particular, the reachability problem for plain VAS that consists in deciding if a semilinear set $\vec{C}_{bad}\subseteq \setN^d$ has an empty intersection with $\poststar{\vec{A}}{\vec{C}_0}$ reduces to the reachability problem of $(\vec{A}',\vec{C}_0')$ for the semilinear set $\vec{C}_{bad}'$ defined as $\vec{C}_{bad}\times\{1\}$. The computation of inductive invariants for plain VAS also reduces to the computation of inductive invariants for periodic VAS as follows. Assume that $\vec{I}'$ is an inductive invariant for the periodic VAS $(\vec{A}',\vec{C}_0')$ such that $\vec{I}'\cap \vec{C}_{bad}'$ is empty. We introduce the semilinear set $\vec{I}$ defined as $\{\vec{x}\in\setN^d \mid (\vec{x},1)\in \vec{I}'\}$. We observe that $\vec{I}$ is an inductive invariant for $(\vec{A},\vec{C}_0)$ such that $\vec{I}\cap \vec{C}_{bad}$ is empty.

\subsection{Periodic Semilinear Inductive Invariants}\label{sub:mainthm}
The main result of the paper is the proof of the following theorem that extends Theorem~\ref{thm:inv} of~\cite{DBLP:conf/lics/Leroux09,DBLP:conf/popl/Leroux11,Turing-100:Vector_Addition_Systems_Reachability} to periodic VAS. We have seen in \cref{ex:BnF} that the back-and-forth construction of~\cite{DBLP:conf/lics/Leroux09,DBLP:conf/popl/Leroux11,Turing-100:Vector_Addition_Systems_Reachability} may produce non-periodic inductive invariants.
\begin{theorem}\label{thm:invper}
  The reachability set of a periodic VAS $(\vec{A},\vec{C}_0)$ has an empty intersection with a semilinear set $\vec{C}_{bad}$ if and only if there exists a periodic semilinear inductive invariant $\vec{I}$ of $(\vec{A},\vec{C}_0)$ such that $\vec{I}\cap \vec{C}_{bad}$ is empty.
\end{theorem}

\section{Geometry of VAS Reachability Sets}\label{sec:geovas}
In this section we recall the classical well-quasi-order on the set of runs of a VAS. We also recall the geometrical characterization of the VAS reachability sets by almost semilinear sets. Finally, we introduce the notion of interior vectors of a periodic set, a new notion that provides a way to push vectors in the approximation of almost linear sets $\vec{L}$ (called linearization in the sequel) to be in the original set $\vec{L}$.

\subsection{Well-Quasi-Orders}
\label{subsec:wqo}
A binary relation $\sqsubseteq$ on a set $S$ is said to be \emph{almost-full}~\cite{VeldmanBezem93} if
for every infinite sequence $(s_n)_{n\in\setN}$ of elements in $S$, there exists $m<n$ such that $s_m\sqsubseteq s_n$.
Let us recall that a relation is almost-full if, and only if, for every infinite sequence $(s_n)_{n\in\setN}$ of elements in $S$, there exists an infinite sequence $n_0<n_1<\cdots$ such that $s_{n_i}\sqsubseteq s_{n_j}$ for all $i<j$.
This characterization is an immediate consequence of the infinite version of Ramsey's Theorem.

A quasi-order that is almost-full is called a \emph{well-quasi-order}. Let us recall that the classical order $\leq$ on $\setN$ is a well-quasi-order. Moreover, since the Cartesian product of two well-quasi-orders is a well-quasi-order (\emph{Dickson's lemma}), it follows that the component-wise extension of $\leq$ on $\setN^d$ is also a well-quasi-order.

In this subsection,
we associate an almost-full relation with each semilinear subset of $\setN^d$.
These almost-full relations will be used in the sequel (namely, in the proof of \cref{lem:n-existence}).

\medskip

Given a set $\vec{S}\subseteq\setN^d$, we introduce the binary relation $\leq_{\vec{S}}$ on $\vec{S}$ defined by $\vec{x}\leq_{\vec{S}}\vec{y}$ if $\vec{x}\leq \vec{y}$ and $\vec{x}+\setN(\vec{y}-\vec{x})\subseteq \vec{S}$.
We observe that $\leq_{\vec{S}}$ is not necessarily transitive
(this is the reason why we need the notion of almost-full relations).

\begin{lemma}
  \label{lem:semilinear-to-well}
  The relation $\leq_{\vec{S}}$ is almost-full for every semilinear set $\vec{S}\subseteq \setN^d$.
\end{lemma}
\begin{proof}
  Let us first prove the lemma for a linear set $\vec{L}$ given by a linear presentation $(\vec{b},\vec{G})$, i.e such that $\vec{L}=\vec{b}+\per{\vec{G}}$ for $\vec{b}\in\setN^d$ and $\vec{G}=\{\vec{g}_1,\ldots,\vec{g}_k\}$ is a finite set of vectors in $\setN^d$. Let us consider a sequence $(\vec{x}_n)_{n\in\setN}$ of vectors in $\vec{L}$. We introduce the mapping $f:\setN^k\rightarrow\vec{L}$ defined by $f(c_1,\ldots,c_k)=\vec{b}+\sum_{j=1}^kc_j\vec{g}_j$. Since $f(\setN^k)=\vec{L}$, we deduce that for every $n\in\setN$, there exists $\vec{c}_n\in\setN^k$ such that $\vec{x}_n=f(\vec{c}_n)$. Since the quasi-order $\leq$ on $\setN^k$ is almost-full, there exists $m<n$ such that $\vec{c}_m\leq\vec{c}_n$. It follows that $\vec{x}_m\leq_{\vec{L}}\vec{x}_n$ and we have proved that $\leq_{\vec{L}}$ is almost-full.

  Now, let us consider a semilinear set $\vec{S}$ of the form $\vec{L}_1\cdots\ldots\cup\vec{L}_k$ where $\vec{L}_j$ is a linear set for every $j$. Let us consider an infinite sequence $(\vec{x}_n)_{n\in\setN}$ in $\vec{S}$. There exists $j$ such that $\vec{x}_n\in\vec{L}_j$ for infinitely many indices $n\in\setN$. As $\leq_{\vec{L}_j}$ is almost-full, we deduce that there exists $m<n$ such that $\vec{x}_m\leq_{\vec{L}_j}\vec{x}_n$. In particular, as $\vec{L}_j\subseteq \vec{S}$ we deduce that $\vec{x}_m\leq_{\vec{S}}\vec{x}_n$. We have proved that $\leq_{\vec{S}}$ is almost-full.
\end{proof}

\medskip

We introduce the binary relation $\unlhd$ on runs of a VAS $\vec{A}$ defined by $\rho\unlhd \rho'$ if $\rho=\vec{c}_0\ldots\vec{c}_k$ for some configurations $\vec{c}_0,\ldots,\vec{c}_k$ and $\rho'=\rho_0\ldots\rho_k$ for some runs $\rho_j$ such that $\src{\rho_j},\tgt{\rho_j}\geq \vec{c}_j$. Let us recall~\cite{JANCARwqo,LS15} that $\unlhd$ is a well-quasi-order satisfying the following \emph{amalgamation property}.
\begin{lemma}[Amalgamation Property\cite{LS15}]
  \label{lem:amalgamation}
  For all runs $\rho\unlhd \rho_1,\rho_2$ of a VAS $\vec{A}$, there exists a run $\tau$ of $\vec{A}$ such that:
  \begin{itemize}
  \item $\rho_1,\rho_2\unlhd \tau$,
  \item $\src{\rho}+\src{\tau}=\src{\rho_1}+\src{\rho_2}$, and
  \item $\tgt{\rho}+\tgt{\tau}=\tgt{\rho_1}+\tgt{\rho_2}$.
  \end{itemize}
\end{lemma}

\subsection{Almost Semilinear Sets} \label{subsec:almost-SL}
Thanks to the well-quasi-order $\unlhd$ on VAS runs, the reachability set of a VAS $(\vec{A},\vec{C}_0)$ can be decomposed as finite union of sets \emph{called almost linear sets} since they are geometrically close to the linear sets.

\medskip

An \emph{almost linear set} $\vec{L}$ is a set of the form $\vec{b}+\vec{P}$ where $\vec{b}\in\setN^d$ and $\vec{P}\subseteq\setN^d$ is a periodic set such that the \emph{cone} $\setQ_{\geq 0}\vec{P}$ is definable in $FO(\setQ_{\geq 0},+)$.\footnote{Since $FO(\setQ_{\geq 0},+,\leq)$ admits quantifier elimination, a set $\vec{X}\subseteq \setQ_{\geq 0}^d$ is definable in this logic if, and only if, it is a boolean combination of \emph{half-spaces} $\{(x_1,\ldots,x_d)\in\setQ_{\geq 0}^d \mid a_1x_1+\cdots+a_dx_d\sim 0\}$ where $a_1,\ldots,a_d\in\setZ$ and $\sim$ is either $\geq$ or $>$.
}

An \emph{almost semilinear set} is a finite union of almost linear sets.
\begin{example}\label{ex:pow}
  A semilinear periodic set $\vec{P}$ is such that the \emph{cone} $\setQ_{\geq 0}\vec{P}$ is definable in $FO(\setQ_{\geq 0},+)$. The converse is not true since the periodic set $\vec{P}=\{(0,0)\}\cup \{(x,y)\in\setN^2\mid x\geq 1\wedge y\leq 2^x\}$ is such that the cone $\setQ_{\geq 0}\vec{P}$ is definable in $FO(\setQ_{\geq 0},+)$ since it is equal to $\{(0,0)\}\cup (\setQ_{> 0}\times \setQ_{\geq 0})$. But $\vec{P}$ is not semilinear.
\end{example}

We recall the following theorem.
\begin{theorem}[\cite{DBLP:conf/popl/Leroux11}]\label{lem:reachalmost}
  The set $\poststar{\vec{A}}{\vec{C}_0}\cap\vec{S}$ is almost semilinear for every semilinear set $\vec{S}\subseteq \setN^d$ and for every VAS $(\vec{A},\vec{C}_0)$. 
\end{theorem}

We also recall some definitions that provide a way to over-approximate almost linear sets by linear sets~(see \cite{DBLP:conf/popl/Leroux11} for more details). The \emph{linearization} of an almost linear set $\vec{L}=\vec{b}+\vec{P}$ is the set $\lin{\vec{L}}=\vec{b}+\vec{Q}$ where $\vec{Q} \coloneqq (\vec{P}-\vec{P})\cap \overline{\setQ_{\geq 0}\vec{P}}$ and $\overline{\vec{X}}$ where $\vec{X}\subseteq \setQ^d$ is the classical \emph{topological closure} of $\vec{X}$.
Let us recall that $\lin{\vec{L}}$ is a linear set since $\vec{Q}$ is a finitely-generated periodic set~\cite[Lemma 5.1]{DBLP:conf/popl/Leroux11}. A \emph{linearization} of an almost semilinear set $\vec{S}$ is a finite set $\{\lin{\vec{L}_1},\ldots,\lin{\vec{L}_k}\}$ where $\vec{L}_1,\ldots,\vec{L}_k$ are almost linear sets satisfying $\vec{S}=\bigcup_{j=1}^k\vec{L}_j$. Since such a decomposition of an almost semilinear set into finite union of almost linear sets is not unique, an almost semilinear set can admit several linearizations.
\begin{example}\label{ex:pow2}
  Let us come back to Example~\ref{ex:pow} and notice that $\overline{\setQ_{\geq 0}\vec{P}}=\setQ_{\geq 0}^2$ and $\vec{P}-\vec{P}=\setZ^2$. It follows that $\lin{\vec{P}}=\setN^2$.
\end{example}

\begin{example}\label{ex:linlin}
  The linearization of the linear set $\vec L=\per{\{2,3\}}$ is the linear set $\setN$. Observe that $\setN$ is strictly larger than $\vec L$ since $1\not\in \vec L$.
\end{example}

Linearizations provide a simple way to over-approximate almost linear sets and almost semilinear sets by linear sets and semilinear sets respectively. The following lemma shows that the space-dimension is unchanged by those approximations.
\begin{lemma} \label{lem:spacelin}
  We have $\sdim{\vec{L}}=\sdim{\lin{\vec{L}}}$ for every almost linear set $\vec{L}$. 
\end{lemma}
\begin{proof}
  Assume that $\vec{L}=\vec{b}+\vec{P}$ and observe that $\lin{\vec{L}}=\vec{b}+\vec{Q}$ where $\vec{Q}$ is the periodic set $(\vec{P}-\vec{P})\cap \overline{\setQ_{\geq 0}\vec{P}}$. Let $\vec{V}$ and $\vec{W}$ be the vector space spanned by $\vec{P}$ and $\vec{Q}$. Corollary~\ref{cor:dimspacedim} shows that $\sdim{\vec{P}}=\dim{\vec{V}}$ and $\sdim{\vec{Q}}=\dim{\vec{W}}$. Since $\sdim{\vec{L}}=\sdim{\vec{P}}$ and $\sdim{\lin{\vec{L}}}=\sdim{\vec{Q}}$, it is sufficient to prove that $\vec{V}=\vec{W}$. As $\vec{P}\subseteq \vec{Q}$ we get $\vec{P}\subseteq \vec{W}$. Hence $\vec{V}\subseteq \vec{W}$ by minimality of $\vec{V}$. For the converse inclusion, the inclusion $\vec{P}\subseteq \vec{V}$ implies that $\vec{P}-\vec{P}\subseteq \vec{V}$ and $\setQ_{\geq 0}\vec{P}\subseteq \vec{V}$. Since a vector-space is topologically-closed, from $\setQ_{\geq 0}\vec{P}\subseteq \vec{V}$ we derive $\overline{\setQ_{\geq 0}\vec{P}}\subseteq \vec{V}$. We have proved $\vec{Q}\subseteq\vec{V}$. It follows that $\vec{W}\subseteq\vec{V}$ by minimality of $\vec{W}$. We have proved that $\vec{V}=\vec{W}$.
\end{proof}

\subsection{Interior Vectors}
We now introduce a new technique to refine the approximation thanks to the notion of \emph{interior vectors of a periodic set}. Formally, an \emph{interior vector} of a periodic set $\vec{P}\subseteq \setN^d$ is a vector $\vec{p}\in\setN^d$ such that for every $\vec{q}\in\vec{P}$, there exists $n\in\setN_{>0}$ satisfying $n\vec{p}\in \vec{q}+\vec{P}$. \emph{Interior vectors} can be geometrically characterized by the topological interior of the cone $\setQ_{\geq 0}\vec{P}$ spanned by $\vec{P}$. Since this characterization is not used in this paper, we just provide the following results about interior vectors%
\ifthenelse{\boolean{arXivLongVersion}}{%
\ (Detailed proofs are given in appendix).
}{.}

The following lemma shows as direct corollary that any periodic set admits an interior vector.
\begin{restatable}{lemma}{lemintcara}\label{lem:intcara}
    A vector $\vec{p}\in\setN^d$ is interior to a periodic set $\vec{P}\subseteq\setN^d$ if, and only if, there exists a finite sequence $\vec{p}_1,\ldots,\vec{p}_k \in \vec P$ of vectors spanning the same vector space as $\vec{P}$ such that:
    $$\vec{p}\in\sum_j \setQ_{>0}\vec{p}_j$$
\end{restatable}

The next lemma shows that interior vectors of $\vec{P}$ can be obtained just by considering interior vectors of $\vec{Q} \coloneqq (\vec{P}-\vec{P})\cap \overline{\setQ_{\geq 0}\vec{P}}$. Notice that if $\vec{Q}$ is finitely-generated, i.e. a periodic set of the form $\vec{Q}=\per{\vec{G}}$ for some finite set $\vec{G}\subseteq \setN^d$, thanks to the previous lemma~\ref{lem:intcara}, we deduce that the vector $\sum_{\vec{g}\in\vec{G}}\vec{g}$ is an interior vector of $\vec{Q}$, and in particular an interior vector of $\vec{P}$.
\begin{restatable}{lemma}{lemintPQ}\label{lem:intPQ}
  The set of interior vectors of a periodic set $\vec{P}\subseteq \setN^d$ coincides with the set of interior vectors of the periodic set $\vec{Q}\coloneqq(\vec{P}-\vec{P})\cap\overline{\setQ_{\geq 0}\vec{P}}$.
\end{restatable}

Interior vectors are used in the sequel to translate a vector $\vec{y}\in\vec{Q}$ in order to obtain an infinite subset of $\vec{P}$ of the form $\vec{y}+r\setN_{>0}\vec{x}$ as shown in the following lemma.
\begin{restatable}{lemma}{lempushin}\label{lem:pushin}
  Let $\vec{P}\subseteq \setN^d$ be a periodic set and
  let $\vec{Q} \coloneqq (\vec{P}-\vec{P}) \cap \overline{\setQ_{\geq 0}\vec{P}}$.
  For every $\vec{y} \in \vec{Q}$, 
  for every interior vector $\vec{q}$ of $\vec{Q}$, and 
  for every $\vec{x} \in \vec{q} + \vec{Q}$,
  there exists $r \in \setN_{> 0}$ such that
  $\vec{y} + r \setN_{> 0} \vec{x} \subseteq \vec{P}$.
\end{restatable}

\section{Forward Semilinear Periodic Inductive Invariants}\label{sec:algorithm}
We are now equipped with the necessary ingredients to prove the main result of the paper,
namely \cref{thm:invper}.
Our proof relies on a construction of the desired semilinear periodic inductive invariant $\vec{I}$
by means of an algorithm with oracle calls.
This algorithm,
dubbed $\mathtt{InductiveInvariant}$,
is defined in \cref{algo:inductive-invariant}.
We don't address whether the oracle calls can be effectively implemented, so the algorithm should be understood only as a convenient presentation of the existence proof.
The oracle call at line~\ref{l:oracle} takes as input a semilinear set $\vec{I}$ (stored in the variable $\mathtt{Inv}$) and returns a linearization of the almost semilinear set $\poststar{\vec A}{\vec I} \setminus \vec I$.
Recall that $\poststar{\vec A}{\vec I} \setminus \vec I$ is almost semilinear by \cref{lem:reachalmost}.
A second oracle call at line~\ref{l:choose_n} selects some vectors $\vec{q}_1, \ldots, \vec{q}_k \in \setN^d$ satisfying a suitable condition.
The existence of these vectors is established in \cref{cor:n-existence}.

\begin{algorithm}[t]
  \caption{$\mathtt{InductiveInvariant}(\vec A, \vec{C}_0, \vec{C}_{\text{bad}})$}
  \label{algo:inductive-invariant}
  \begin{algorithmic}[1]
    \Require
    A periodic $d$-VAS $(\vec A, \vec{C}_0)$ and a semilinear set $\vec{C}_{\text{bad}} \subseteq \setN^d \setminus \poststar{\vec A}{\vec{C}_0}$.
    \Ensure
    A semilinear periodic inductive invariant of $(\vec{A}, \vec{C}_0)$ that is disjoint from $\vec{C}_{bad}$.
    \State $\mathtt{Inv} \gets \per{\vec{C}_0}$
    \While{$\poststar{\vec A}{\mathtt{Inv}} \not\subseteq \mathtt{Inv}$}
    \label{while:start}
    \LComment{%
      The existence of $\{\vec{L}_1, \ldots, \vec{L}_k\}$ at line \ref{l:oracle} is ensured by \cref{lem:reachalmost}
    }
    \State Let $\{\vec{L}_1, \ldots, \vec{L}_k\}$ be a linearization of $\poststar{\vec A}{\mathtt{Inv}} \setminus \mathtt{Inv}$
    \label{l:oracle}
    \LComment{%
      The existence of $\vec{q}_1, \ldots, \vec{q}_k$ at line~\ref{l:choose_n} is ensured by \cref{cor:n-existence}
    }
    \State Let $\vec{q}_1, \ldots, \vec{q}_k \in \setN^d$ such that
    $\vec{L}'_i \coloneqq (\vec{L}_i + \vec{q}_i) \subseteq \vec{L}_i$ for each $i \in \{1, \ldots, k\}$ and
    \makebox[5mm]{}%
    $\poststar{\vec A}{\per{\mathtt{Inv} \cup \vec{L}'_1 \cup \cdots \cup \vec{L}'_k}} \subseteq (\mathtt{Inv} \cup \vec{L}_1 \cup \cdots \cup \vec{L}_k) \setminus \vec{C}_{\text{bad}}$
    \label{l:choose_n}
    \State $\mathtt{Inv} \gets \per{\mathtt{Inv} \cup \vec{L}'_1 \cup \cdots \cup \vec{L}'_k}$
    \EndWhile
    \label{while:end}
    \State \Return $\mathtt{Inv}$
  \end{algorithmic}
\end{algorithm}

\smallskip

Termination and correctness of this algorithm,
under the assumption that we can always find $\vec{q}_1, \ldots, \vec{q}_k$ at line \ref{l:choose_n},
come from these two lemmas.
We assume for the remainder of this section that
$(\vec A, \vec{C}_0)$ is a periodic $d$-VAS and that $\vec{C}_{\text{bad}} \subseteq \setN^d$ is a semilinear set that is disjoint from $\poststar{\vec A}{\vec{C}_0}$.
Consider an execution of $\mathtt{InductiveInvariant}(\vec A, \vec{C}_0, \vec{C}_{\text{bad}})$.
As an immediate consequence of the two following claims,
we get that the execution terminates and returns
a semilinear periodic inductive invariant of $(\vec{A}, \vec{C}_0)$ that is disjoint from $\vec{C}_{bad}$.

\begin{claim}[Invariant] \label{cl:invariant}
  At the beginning and end of each iteration of the \textbf{while}-loop, $\mathtt{Inv}$ is a semilinear periodic set and $\poststar{\vec A}{\mathtt{Inv}}$ is disjoint from $\vec{C}_{\text{bad}}$.
\end{claim}

\begin{claim}[Termination] \label{cl:termination}
  Each iteration of the \textbf{while}-loop strictly decreases the space-dimension of $\poststar{\vec A}{\mathtt{Inv}} \setminus \mathtt{Inv}$.
\end{claim}

\begin{proof}
  Consider an iteration of the \textbf{while}-loop.
  Let $\vec{I}$ be the value of $\mathtt{Inv}$ at the beginning (line~\ref{while:start}) and
  let $\vec{J}$ be its value at the end (line~\ref{while:end}).
  We need to prove that $\sdim{\poststar{\vec A}{\vec{J}} \setminus \vec{J}} < \sdim{\poststar{\vec A}{\vec{I}} \setminus \vec{I}}$. 

  \smallskip

  According to the body of the \textbf{while}-loop,
  there exist
  a linearization $\{\vec{L}_1, \ldots, \vec{L}_k\}$ of $\poststar{\vec A}{\vec{I}} \setminus \vec{I}$ and
  $\vec{q}_1, \ldots, \vec{q}_k \in \setN^d$ with
  $(\vec{L}_i + \vec{q}_i) \subseteq \vec{L}_i$ for each $i \in \{1, \ldots, k\}$ such that,
  firstly,
  $\vec{J} = \per{\vec{I} \cup \vec{L}'_1 \cup \cdots \cup \vec{L}'_k}$ where
  $\vec{L}'_i \coloneqq (\vec{L}_i + \vec{q}_i)$ for each $i$, and secondly,
  $\poststar{\vec A}{\vec{J}} \subseteq (\vec{I} \cup \vec{L}_1 \cup \cdots \cup \vec{L}_k)$.
  By definition of linearizations of almost semilinear sets,
  there exist almost linear sets $\vec{H}_1, \ldots, \vec{H}_k$ such that
  $\poststar{\vec A}{\vec{I}} \setminus \vec{I} = \bigcup_{i=1}^k \vec{H}_i$ and
  $\vec{L}_i = \lin{\vec{H}_i}$ for each $i \in \{1, \ldots, k\}$.
  We derive from \cref{lem:spacelin} that
  $\sdim{\poststar{\vec A}{\vec{I}} \setminus \vec{I}} = \max_{i} \sdim{\vec{H}_i} = \max_{i} \sdim{\vec{L}_i}$.
  Let us now upper-bound $\sdim{\poststar{\vec A}{\vec{J}} \setminus \vec{J}}$.
  Observe that $\vec{J} \supseteq (\vec{I} \cup \vec{L}'_1 \cup \cdots \cup \vec{L}'_k)$ and
  recall that $\poststar{\vec A}{\vec{J}} \subseteq (\vec{I} \cup \vec{L}_1 \cup \cdots \cup \vec{L}_k)$.
  It follows that
  $\poststar{\vec A}{\vec{J}} \setminus \vec{J}$ is contained in
  $(\vec{L}_1 \setminus \vec{L}'_1) \cup \cdots \cup (\vec{L}_k \setminus \vec{L}'_k)$,
  hence,
  $\sdim{\poststar{\vec A}{\vec{J}} \setminus \vec{J}} \leq \max_{i} \sdim{\vec{L}_i \setminus \vec{L}'_i}$.
  By \cref{lem:spacediff:linear},
  $\sdim{\vec{L}_i \setminus \vec{L}'_i} < \sdim{\vec{L}_i}$
  for each $i \in \{1, \ldots, k\}$.
  We derive that
  $\sdim{\poststar{\vec A}{\vec{J}} \setminus \vec{J}} < \max_{i} \sdim{\vec{L}_i} = \sdim{\poststar{\vec A}{\vec{I}} \setminus \vec{I}}$.
\end{proof}

The remainder of this section is devoted to the proof that some $\vec{q}_1, \ldots, \vec{q}_k \in \setN^d$ satisfying the condition at line~\ref{l:choose_n} always exist (see \cref{cor:n-existence}).
We will need the following easy consequence of \cref{lem:amalgamation}.

\begin{lemma}
  \label{lem:amalgamation:iterated}
  For every runs $\rho \unlhd \rho'$ of a VAS $\vec{A}$ and for every $r \in \setN$,
  there exists a run $\tau$ of $\vec{A}$ such that
  $\src{\tau} = \src{\rho} + r(\src{\rho'} - \src{\rho})$, and
  $\tgt{\tau} = \tgt{\rho} + r(\tgt{\rho'} - \tgt{\rho})$.
\end{lemma}

\begin{lemma} \label{lem:n-existence}
  Let $\vec{I} \subseteq \setN^d$ be a semilinear periodic set and
  consider a linearization $\{\vec{L}_1, \ldots, \vec{L}_k\}$ of the almost semilinear set $\poststar{\vec A}{\vec{I}} \setminus \vec{I}$.
  Let $(\vec{b}_i, \vec{G}_i)$ be a linear-presentation of $\vec{L}_i$ for each $i \in \{1, \ldots, k\}$.
  For every semilinear set $\vec{T} \subseteq \setN^d$ such that $\poststar{\vec A}{\vec{I}} \subseteq \vec{T}$,
  there exists $n \in \setN$ such that
  $\poststar{\vec A}{\per{\vec{I} \cup \vec{L}'_1 \cup \cdots \cup \vec{L}'_k}} \subseteq \vec{T}$
  where
  $\vec{L}'_i \coloneqq (\vec{L}_i + \vec{q}_i)$ and $\vec{q}_i \coloneqq n \sum_{\vec g \in \vec{G}_i} \vec{g}$ for each $i \in \{1, \ldots, k\}$.
\end{lemma}

\begin{proof}
  Consider a semilinear set $\vec{T} \subseteq \setN^d$ such that $\poststar{\vec A}{\vec{I}} \subseteq \vec{T}$.
  For each $n \in \setN$,
  define $\vec{S}_n \coloneqq \per{\vec{I} \cup \vec{L}'_{1,n} \cup \cdots \cup \vec{L}'_{k,n}}$ where
  $\vec{L}'_{i,n} \coloneqq (\vec{L}_i + \vec{q}_{i,n})$ and $\vec{q}_{i,n} \coloneqq n \sum_{\vec g \in \vec{G}_i} \vec{g}$ for each $i \in \{1, \ldots, k\}$.
  Suppose by contradiction that
  $\poststar{\vec A}{\vec{S}_n} \not\subseteq \vec{T}$ for every $n \in \setN$.

  \smallskip

  Define $\vec{Q}_i \coloneqq \per{\vec{G}_i}$ and
  note that $\vec{L}_i = \vec{b}_i + \vec{Q}_i$ for each $i \in \{1, \ldots, k\}$.
  Since $\{\vec{L}_1, \ldots, \vec{L}_k\}$ is a linearization of $\poststar{\vec A}{\vec{I}} \setminus \vec{I}$,
  there exist periodic sets $\vec{P}_1, \ldots, \vec{P}_k \subseteq \setN^d$ such that
  $\poststar{\vec A}{\vec{I}} \setminus \vec{I} = \bigcup_{i=1}^k \vec{b}_i + \vec{P}_i$ and
  $\vec{Q}_i = (\vec{P}_i - \vec{P}_i) \cap \overline{\setQ_{\geq 0}\vec{P}_i}$ for each $i \in \{1, \ldots, k\}$.

  \smallskip

  By hypothesis, for every $n \in \setN$,
  there is a run $\rho_n$ from some $\vec{s}_n \in \vec{S}_n$ to some $\vec{t}_n \in (\setN^d \setminus \vec{T})$.
  Moreover,
  according to the definition of $\vec{S}_n$,
  we can write each $\vec{s}_n$ under the form
  $\vec{s}_n = \vec{i}_n + \sum_{j \in J_n} \left( \mu_{j,n} (\vec{b}_j + \vec{q}_{j,n}) + \vec{u}_{j,n} \right)$
  for
  some $\vec{i}_n \in \per{\vec{I}} = \vec{I}$,
  some subset $J_n \subseteq \{1, \ldots, k\}$ and,
  for each $j \in J_n$,
  some $\mu_{j,n} \in \setN$ with $\mu_{j,n} > 0$ and
  some $\vec{u}_{j,n} \in \per{\vec{G}_j}$.
  Extracting a subsequence if necessary,
  we may assume that $J_m = J_n$ for all $m, n \in \setN$,
  and we let $J$ denote their common value.

  \smallskip

  As $\vec{I}$ and $\setN^d \setminus \vec{T}$ are semilinear,
  the relations $\leq_{\vec{I}}$ and $\leq_{\setN^d \setminus \vec{T}}$ are almost-full
  (see \cref{lem:semilinear-to-well}).
  The usual order $\leq$ on $\setN$ is a well-quasi-order,
  and so is the relation $\unlhd$ on runs of $\vec{A}$
  (see \cref{subsec:wqo}).
  It is also well-known that for every finite subset $\vec{G} \subseteq \setN^d$,
  the binary relation $\preceq_{\vec{G}}$ on $\per{\vec{G}}$ defined by
  $\vec{u} \preceq_{\vec{G}} \vec{v}$ if $(\vec{v} - \vec{u}) \in \per{\vec{G}}$,
  is a well-quasi-order.
  So we can find an increasing pair $m, n \in \setN$ with $m < n$ such that
  $\vec i_m \leq_{\vec{I}} \vec i_n$,
  $\vec t_m \leq_{\setN^d \setminus \vec{T}} \vec t_n$,
  $\mu_{j,m} \leq \mu_{j,n}$,
  $\rho_m \trianglelefteq \rho_n$ and
  $\vec{u}_{j,m} \preceq_{\vec{G}_j} \vec{u}_{j,n}$ for each $j \in J$.
  Let us define,
  $\vec{y}_j \coloneqq \mu_{j,m} \vec{q}_{j,m} + \vec{u}_{j,m}$ and
  $\vec{x}_j \coloneqq \mu_{j,n} \vec{q}_{j,n} - \mu_{j,m} \vec{q}_{j,m} + \vec{u}_{j,n} - \vec{u}_{j,m}$,
  for each $j \in J$.
  Observe that $\vec{y}_j \in \vec{Q}_j$ and
  $\vec{x}_j \in (\nu \sum_{\vec{g} \in \vec{G}_j} \vec{g}) + \vec{Q}_j$,
  where $\nu \coloneqq n\mu_{j,n} - m\mu_{j,m} > 0$.
  We derive from \cref{lem:intcara,lem:pushin} that
  there exists $r_j \in \setN_{> 0}$ such that
  $\vec{y}_j + r_j \setN_{> 0} \vec{x}_j \subseteq \vec{P}_j$.
  By defining $r \coloneqq \prod_{j \in J} r_j$,
  we get that $(\vec{y}_j + r \vec{x}_j) \in \vec{P}_j$ for each $j \in J$.

  \smallskip

  Recall that $\rho_m \trianglelefteq \rho_n$.
  By \cref{lem:amalgamation:iterated},
  there exists a run $\tau$ of $\vec{A}$ such that
  $\src{\tau} = \vec{s}_m + r(\vec{s}_n - \vec{s}_m)$, and
  $\tgt{\tau} = \vec{t}_m + r(\vec{t}_n - \vec{t}_m)$.
  Observe that $\tgt{\tau} \in (\setN^d \setminus \vec{T})$ since
  $\vec{t}_m, \vec{t}_n \in (\setN^d \setminus \vec{T})$ and
  $\vec t_m \leq_{\setN^d \setminus \vec{T}} \vec t_n$.
  So to get a contradiction,
  there only remains to show that $\src{\tau}$ is in $\poststar{\vec A}{\vec{I}}$.
  Indeed, if this is the case, by prepending to $\tau$ a run from $\vec{I}$ to $\src{\tau}$,
  we get a run from $\vec{I}$ to $\setN^d \setminus \vec{T}$,
  contradicting the assumption that
  $\poststar{\vec A}{\vec{I}} \subseteq \vec{T}$.

  \smallskip

  We now show that $\src{\tau}$ is in $\poststar{\vec A}{\vec{I}}$.
  Notice that $\poststar{\vec A}{\vec{I}}$ is periodic since $\vec{I}$ is periodic
  (see \cref{lem:periodic-vas}).
  So it is enough to express $\src{\tau}$ as a sum of elements of $\poststar{\vec A}{\vec{I}}$.
  We have
  $\vec{s}_m = \vec{i}_m + \sum_{j \in J} \left( \mu_{j,m} \vec{b}_j + \vec{y}_j \right)$
  and
  $\vec{s}_n - \vec{s}_m = \vec{i}_n - \vec{i}_m + \sum_{j \in J} (\mu_{j,n} - \mu_{j,m})\vec{b}_j + \sum_{j \in J} \vec{x}_j$.
  It follows that
  $$
  \src{\tau} =
  \vec{i}_m + r(\vec{i}_n - \vec{i}_m) +
  \sum_{j \in J} \left(\mu_{j,m} - 1 + r(\mu_{j,n} - \mu_{j,m}) \right) \vec{b}_j +
  \sum_{j \in J} \left( \vec{b}_j + \vec{y}_j + r \vec{x}_j \right)
  \ .
  $$
  The term $\vec{i}_m + r(\vec{i}_n - \vec{i}_m)$ belongs to $\vec{I} \subseteq \poststar{\vec A}{\vec{I}}$
  since $\vec{i}_m, \vec{i}_n \in \vec{I}$ and
  $\vec i_m \leq_{\vec{I}} \vec i_n$.
  For each $j \in J$,
  we have
  $\mu_{j,m} \geq 1$, $\mu_{j,n} \geq \mu_{j,m}$ and $\vec{b}_j \in \poststar{\vec A}{\vec{I}}$,
  hence,
  by periodicity,
  $\left(\mu_{j,m} - 1 + r(\mu_{j,n} - \mu_{j,m}) \right) \vec{b}_j$ is in $\poststar{\vec A}{\vec{I}}$.
  Finally,
  the integer $r$ was chosen so that
  $(\vec{y}_j + r \vec{x}_j) \in \vec{P}_j$ for each $j \in J$,
  hence,
  $\vec{b}_j + \vec{y}_j + r \vec{x}_j$ is in $\vec{b}_j + \vec{P}_j \subseteq \poststar{\vec A}{\vec{I}}$.
  We have written $\src{\tau}$ as a sum of elements of $\poststar{\vec A}{\vec{I}}$,
  which completes the proof.
\end{proof}

\begin{corollary} \label{cor:n-existence}
  In \cref{algo:inductive-invariant}, one can always find $\vec{q}_1, \ldots, \vec{q}_k$ at line~\ref{l:choose_n}.
\end{corollary}

\begin{proof}
  By \cref{cl:invariant},
  the value of $\mathtt{Inv}$ at line~\ref{l:choose_n} is a semilinear periodic set and
  $\poststar{\vec A}{\mathtt{Inv}}$ is disjoint from $\vec{C}_{\text{bad}}$.
  The corollary follows from \cref{lem:n-existence} instantiated with the semilinear set
  $\vec{T} \coloneqq (\mathtt{Inv} \cup \vec{L}_1 \cup \cdots \cup \vec{L}_k) \setminus \vec{C}_{\text{bad}}$.
\end{proof}

\begin{remark} \label{rk:2variants}
  As already observed in Section~\ref{sub:plain2per}, semilinear inductive invariants for a plain VAS $(\vec{A},\vec{C}_0)$ and a set $\vec{C}_{\text{bad}}$
  can be derived from semilinear inductive invariants for a periodic VAS $(\vec{A}',\vec{C}_0')$ and a set $\vec{C}'_{\text{bad}}$ obtained by introducing an extra counter that is untouched by $\vec{A}'$. Instead of applying \cref{algo:inductive-invariant} to $(\vec{A}',\vec{C}_0')$ and $\vec{C}'_{\text{bad}}$, one may apply a variant of that algorithm obtained by just removing the operator $\operatorname{Per}$ each time it is used. In fact, notice that if $\vec{X}, \vec{Y}, \vec{Z}$ are subsets of $\setN^d$, and if $\vec{X}'\coloneqq\per{\vec{X}\times\{1\}}$, $\vec{Y}'\coloneqq\per{\vec{Y}\times\{1\}}$, and $\vec{Z}'\coloneqq\per{\vec{Z}\times\{1\}}$, then $\vec{Z}=\vec{X}\cup\vec{Y}$ if, and only if, $\vec{Z}'=\vec{X}'+\vec{Y}'$.
\end{remark}

\section{Conclusion}\label{sec:conclusion}
The recent approach of Leroux \cite{DBLP:conf/popl/Leroux11,Turing-100:Vector_Addition_Systems_Reachability}
to solve the VAS reachability problem is based on the fundamental property
(recalled in \cref{thm:inv}) that
every semilinear set $\vec{S}$ containing the reachability set of a VAS $(\vec{A}, \vec{C}_0)$ also contains a semilinear inductive invariant $\vec{I}$ for $(\vec{A}, \vec{C}_0)$.
Let us call this property the ``semilinear inductive invariant'' property.

In this paper,
we have shown that $\vec{I}$ can in addition be required to be periodic when the VAS is periodic. Moreover, our method builds invariants solely from the set $\vec{C}_0$ of initial configurations and avoids the need for backward reasoning. We believe that such a forward approach is more suitable than the back-and-forth approach for attacking the branching VAS reachability problem, which is still open.
\ifthenelse{\boolean{arXivLongVersion}}{%
We discuss potential extensions of the ``semilinear inductive invariant'' property in appendix (see \cref{sec:appliBGVAS}).
}{}

\bibliographystyle{plainurl}
\bibliography{biblio.bib}

\ifthenelse{\boolean{arXivLongVersion}}{%
\clearpage
\appendix

\section{Detailed presentation of the examples of Section~\ref{sec:example}} \label{app:example}

\subsection{Reminder of Definitions from \cref{sec:preliminaries}}

The only notions needed to read this appendix are recalled below.

The \emph{dimension} of a set $\vec X \subseteq \setQ^d$ is the smallest integer $d'$ such that $\vec X$ is included in a finite union of affine subspaces of dimension $d'$.

A \emph{periodic set} (of $\setN^d$) is a set that contains the zero vector and is closed under addition. For every $\vec X \subseteq \setN^d$, we denote by $\per{\vec X}$ the periodic set generated by $\vec X$, i.e., the set of all finite sums of elements of $\vec X$.

A \emph{semilinear} set is a finite union of \emph{linear} sets, i.e., sets of the form $\vec b + \per{\{\vec p_1,\ldots,\vec p_k\}}$ for some $k \in \setN$ and $\vec b, \vec p_1,\ldots,\vec p_k \in \setN^d$.  
Semilinear sets are closed under union, intersection, and complement.

A VAS of dimension $d$ consists of a finite set of vectors $\vec A \subseteq \setZ^d$ together with a semilinear set $\vec C_0 \subseteq \setN^d$ of initial configurations.

A \emph{run} (from $\vec c_0$ to $\vec c_\ell$) is a sequence $\vec c_0,\ldots,\vec c_\ell$ in $\setN^d$ such that for all $i \in [0,\ell-1]$, there exists $\vec a_i \in \vec A$ with $\vec c_{i+1} = \vec c_i + \vec a_i$.

For any $\vec C \subseteq \setN^d$, we write $\poststar{\vec A}{\vec C}$ for the set of targets of runs starting in $\vec C$, and $\prestar{\vec A}{\vec C}$ for the set of sources of runs ending in $\vec C$.

An \emph{inductive invariant} is a set $\vec I$ such that $\vec C_0 \subseteq \vec I$ and $\poststar{\vec A}{\vec I} = \vec I$.

\subsection{Intuitive Description of the Two Constructions}

Leroux proved that VAS reachability sets admit structured decompositions that can be approximated by semilinear sets. More precisely, for all semilinear sets $\vec S$ and $\vec T$, the sets $\poststar{\vec A}{\vec S} \cap \vec T$ and $\prestar{\vec A}{\vec S} \cap \vec T$ can be decomposed into finitely many components (called \emph{almost linear sets}) of the form $\vec b + \vec P$, where $\vec P \subseteq \setN^d$ is a periodic set satisfying suitable geometric properties.

Each such component admits a semilinear over-approximation, called its \emph{linearization}, defined as
$$
\lin{\vec P} \coloneqq (\vec P - \vec P) \cap \overline{\mathbb{Q}_+ \vec P}.
$$
A decomposition into such almost linear sets is called an \emph{almost semilinear decomposition}, and sets admitting such decompositions are called \emph{almost semilinear}. See \cref{subsec:almost-SL} for a formal definition of almost semilinear sets.

Using these decompositions and their semilinear approximations, Leroux showed that non-reachability in a VAS is always witnessed by a semilinear inductive invariant: for every semilinear set $\vec C_{\text{bad}}$ such that $\vec C_{\text{bad}} \cap \poststar{\vec A}{\vec C_0} = \emptyset$, there exists a semilinear inductive invariant $\vec I$ such that $\vec C_{\text{bad}} \cap \vec I = \emptyset$.

This existence proof naturally takes the form of an algorithm (albeit with calls to an undecidable oracle computing linearizations), and we therefore view it as a \emph{construction}. We refer to Leroux’s construction as \emph{back-and-forth}, and to our alternative as \emph{forward-only}.

\medskip

\noindent \textbf{Back-and-forth construction.} Let us first recall the back-and-forth construction. One iteratively grows two semilinear sets $\vec S$ and $\vec T$ such that there is no run from $\vec S$ to $\vec T$. Initially, $\vec S = \vec C_0$ and $\vec T = \vec C_{\text{bad}}$. At termination, the sets $\vec S$ and $\vec T$ are complementary, so $\vec S$ is an inductive invariant.

Each iteration proceeds in two symmetric steps:
$$
\vec T \gets \vec T \cup \setN^d \setminus \bigl(\vec S \cup \lin{\poststar{\vec A}{\vec S} \setminus \vec S}\bigr), 
\qquad
\vec S \gets \vec S \cup \setN^d \setminus \bigl(\vec T \cup \lin{\prestar{\vec A}{\vec T} \setminus \vec T}\bigr)
\footnote{It is not immediate without entering the geometric details, but it is crucial to add to $\vec T$ the complement of $\vec S \cup \lin{\poststar{\vec A}{\vec S} \setminus \vec S}$, rather than simply the complement of $\lin{\poststar{\vec A}{\vec S}}$. Indeed, the set $\poststar{\vec A}{\vec S} \setminus \vec S$ decreases in dimension at each iteration, making the successive linearizations increasingly precise.}.
$$

Intuitively, each step enlarges one side by adding all configurations that are safely outside an over-approximation of what can be reached from the other side.

Moreover, each iteration adds a large set of configurations to $\vec S \cup \vec T$. Indeed, since there is no run from $\vec S$ to $\vec T$, we have $\poststar{\vec A}{\vec S} \cap \prestar{\vec A}{\vec T} = \emptyset$, and hence
$$
\setN^d = (\setN^d \setminus \poststar{\vec A}{\vec S}) \cup (\setN^d \setminus \prestar{\vec A}{\vec T}).
$$
An iteration adds semilinear under-approximations of these two regions to $\vec T$ and $\vec S$, respectively. Because linearizations are sufficiently precise, most of the remaining configurations are covered at each step. Formally, one can show that the dimension of $\setN^d \setminus (\vec S \cup \vec T)$ strictly decreases at each iteration, which guarantees termination.

\medskip

Observe how the back-and-forth construction addresses a key difficulty: reachability sets of VAS can be over-approximated by semilinear sets, but not under-approximated.

A \emph{naive construction} would be to grow a single set $\vec S$, initially equal to $\vec C_0$, by iterating
$$
\vec S \gets \vec S \cup \lin{\poststar{\vec A}{\vec S} \setminus \vec S}.
$$
However, this approach is not sound: $\poststar{\vec A}{\vec S} \cap \vec C_{\text{bad}} = \emptyset$ does not imply $\lin{\poststar{\vec A}{\vec S}} \cap \vec C_{\text{bad}} = \emptyset$.

Ensuring termination is also challenging. At each iteration, $\vec S$ is enlarged not only by configurations that are actually reachable, but also by additional ones introduced by the over-approximation, which may in turn reach new configurations outside $\vec S$. Even if these additional configurations lie in a low-dimensional region, the set of configurations they can subsequently reach may have larger dimension.

\medskip

\noindent \textbf{Forward-only construction.} Our forward-only construction also grows a single set $\vec S$, but avoids these pitfalls. The main technical ingredient is the following lemma: for every semilinear set $\vec S$, every almost semilinear component $\vec b + \vec P \subseteq \poststar{\vec A}{\vec S} \setminus \vec S$, and every semilinear set $\vec T$ containing $\poststar{\vec A}{\vec S} \setminus \vec S$, there exists $\vec p \in \vec P$ such that
$$
\poststar{\vec A}{\vec b + \vec p + \lin{\vec P}} \subseteq \vec T.
$$

One iteration applies this lemma to each component $\vec b + \vec P$ of an almost semilinear decomposition of $\poststar{\vec A}{\vec S} \setminus \vec S$, taking
$$
\vec T = \vec S \cup \vec L \setminus \vec C_{\text{bad}},
$$
where $\vec L$ is a linearization of $\poststar{\vec A}{\vec S} \setminus \vec S$.

The set $\vec T$ enforces two constraints: it prevents $\vec S$ from reaching $\vec C_{\text{bad}}$, and it avoids adding too many superfluous vectors outside $\poststar{\vec A}{\vec S}$. One then adds sets of the form $\vec b + \vec p + \lin{\vec P}$ to $\vec S$.

Note that the added region is neither a subset nor a superset of $\poststar{\vec A}{\vec S} \setminus \vec S$: it is an over-approximation from which a lower-dimensional boundary region has been removed. One can show that the dimension of $\poststar{\vec A}{\vec S} \setminus \vec S$ decreases at each step.

\subsection{Step-by-Step Execution on the Examples}

Before turning to the examples, let us note that both constructions are inherently non-deterministic.

A first source of non-determinism, common to both constructions, is that an almost semilinear set may admit several almost semilinear decompositions (even though VAS reachability sets admit a canonical decomposition based on minimal runs for a well-quasi-order). Different choices may lead to different linearizations, so the notation $\lin{\poststar{\vec A}{\vec S} \setminus \vec S}$ is slightly abusive.

In addition, the forward-only construction involves a second source of non-determinism, namely the choice of the vectors $\vec p$.

\medskip

Also note that we slightly simplified the presentation of the forward-only construction. 
In fact, there are two closely related variants (see~\cref{rk:2variants}). 
One is specific to periodic VAS: it relies on periodicity assumptions and always produces periodic invariants, but does not extend to general VAS. 
The other one, presented here, applies to all VAS, but may yield non-periodic invariants even for periodic systems.

The periodic variant can be recovered by applying $\per{\cdot}$ to $\vec S$ initially and after each update, together with a strengthened form of the lemma (see \cref{lem:n-existence}).
The VAS given in the examples below are both periodic, but they are simple enough, so the two variants actually coincide on them.

\begin{example}
Consider the VAS $(\vec A, \vec C_0)$ given by $\vec A \coloneqq \{(0,1,0), (0,-1,1), (1,2,-2)\}$ and $\vec C_0 = \{\vec 0\}$.

In this system, the last two counters can be viewed as encoding a control state. As long as their sum is less than two, the first counter cannot be modified. Once their sum reaches two (after applying the first action twice), the actions $(0,-1,1)$ and $(1,2,-2)$ allow one to increase the first counter arbitrarily. Moreover, this sum cannot decrease.

We consider unreachable targets of the form $\vec C_{\text{bad}} \coloneqq \{(x_{\text{bad}},1,0)\}$ for some $x_{\text{bad}} \ge 1$.

\medskip
\noindent\textbf{Back-and-forth construction.}
Initially, $\vec S = \vec C_0 = \{(0,0,0)\}$ and $\vec T = \vec C_{\text{bad}}$.

The canonical almost semilinear decomposition of $\poststar{\vec A}{\vec S} \setminus \vec S$ is
$$
((0,1,0) + \vec P) \;\cup\; ((0,0,1) + \vec P),
$$
where
$$
\vec P = \{(0,0,0)\} \cup (\setN \times [1,\infty) \times \setN) \cup (\setN^2 \times [1,\infty)).
$$
Its linearization is $\lin{\vec P} = \setN^3$.

At the first iteration, we update
$$
\vec T \gets \vec T \cup \bigl([1,\infty) \times \{(0,0)\}\bigr).
$$
Moreover,
$$
\prestar{\vec A}{\vec T} \setminus \vec T = \{(x_{\text{bad}},0,0)\},
$$
and since singletons coincide with their linearization, we update
$$
\vec S \gets \vec S \cup \bigl(\setN^3 \setminus (\vec T \cup \{(x_{\text{bad}},0,0)\})\bigr).
$$

At the second iteration, the point $(x_{\text{bad}},0,0)$ is added to $\vec T$, making $\vec S$ and $\vec T$ complementary. The construction terminates and returns the inductive invariant
$$
\vec S = \setN^3 \setminus \bigl([1,\infty) \times \{(0,0)\} \;\cup\; \{(x_{\text{bad}},1,0), (x_{\text{bad}},0,0)\}\bigr).
$$
Observe that this invariant depends strongly on the choice of $x_{\text{bad}}$.

\medskip
\noindent\textbf{Forward-only construction.}
Again, we start with $\vec S = \{\vec (0,0,0)\}$ and a linearization of $\poststar{\vec A}{\vec S} \setminus \vec S$ is
$$
((0,1,0) + \setN^3) \;\cup\; ((0,0,1) + \setN^3).
$$

At the first iteration, we add to $\vec S$ sets of the form $(0,1,0) + \vec p + \setN^3$ and $(0,0,1) + \vec q + \setN^3$, with $\vec p,\vec q \in \setN^3$, chosen so that
$$
\poststar{\vec A}{\{(0,0,0)\} \cup ((0,1,0) + \vec p + \setN^3) \cup ((0,0,1) + \vec q + \setN^3)}
\subseteq ((0,1,0) + \setN^3) \cup ((0,0,1) + \setN^3) \setminus \{(x_{\text{bad}},1,0)\}.
$$
One can take $\vec p = (0,1,0)$ and $\vec q = \vec 0$.

At the second iteration,
$$
\poststar{\vec A}{\vec S} \setminus \vec S = \{(0,1,0)\},
$$
so we simply add this point to $\vec S$. The construction returns the inductive invariant %
$$
\vec S = \{(0,0,0), (0,1,0)\} \cup ((0,2,0) + \setN^3) \cup ((0,0,1) + \setN^3),
$$
which no longer depends on $x_{\text{bad}}$.
\end{example}

\begin{example}
We now give a simple example of a periodic VAS for which the back-and-forth construction does not always yield a periodic inductive invariant.

Consider the VAS $(\vec A, \vec C_0)$ given by $\vec A \coloneqq \{5,6\}$ and $\vec C_0 \coloneqq \{0\}$, and let $\vec C_{\text{bad}} = \{14\}$.

\medskip
\noindent\textbf{Back-and-forth construction.}
Initially, $\vec S = \{0\}$ and $\vec T = \{14\}$.

The canonical almost semilinear decomposition of $\poststar{\vec A}{\vec S} \setminus \vec S$ is
$$
(5 + \per{\{5,6\}}) \;\cup\; (6 + \per{\{5,6\}}),
$$
whose linearization is
$$
(5 + \setN) \;\cup\; (6 + \setN).
$$
Hence, we add $\{1,2,3,4\}$ to $\vec T$.

Next,
$$
\prestar{\vec A}{\vec T} \setminus \vec T = \{8,9\},
$$
which is finite, hence equal to its own linearization. We therefore add $\setN \setminus \{8,9\}$ to $\vec S$.

At the second iteration, the set $\{8,9\}$ is added to $\vec T$, making $\vec S$ and $\vec T$ complementary. The construction returns
$$
\setN \setminus \{1,2,3,4,8,9,14\},
$$
which is not periodic, since it contains $7$ but not $7+7=14$.

\medskip
\noindent\textbf{Forward-only construction.}
We again start with $\vec S = \{0\}$, and a linearization of $\poststar{\vec A}{\vec S} \setminus \vec S$ is $(5 + \setN) \cup (6 + \setN)$.

At the first iteration, we add a set of the form $5 + p + \setN$, with $p \in \setN$, such that
$$
\poststar{\vec A}{\{0\} \cup (5 + p + \setN)} \subseteq (5 + \setN) \setminus \{14\}.
$$
(Adding a set of the form $6 + q + \setN$ is unnecessary, as one such set is included in $5 + p + \setN$.)

One can take $p = 10$, yielding
$$
\vec S = \{0\} \cup (15 + \setN).
$$

At the second iteration, $\poststar{\vec A}{\vec S} \setminus \vec S$ is finite. The finitely many points in $\poststar{\vec A}{\vec S} \setminus \vec S$ are added to $\vec S$, and the construction returns the periodic inductive invariant
$$
\vec S = \{0,5,6,10,11,12\} \cup (15 + \setN) = \setN \setminus \{1,2,3,4,7,8,9,13,14\}.
$$
\end{example}

\section{Proofs of Section~\ref{sec:preliminaries}}\label{app:preliminaries}
\leminsecable*
\begin{proof}
   The proof is performed by induction on $k\geq 1$. The lemma is trivial when $k=1$. Now, assume the lemma proved for some $k\geq 1$ and let $\vec{P}\subseteq \setQ^d$ be a set of vectors such that $\vec{P}+\vec{P}\subseteq \vec{P}$, let $\vec{V}_1,\ldots,\vec{V}_{k+1}$ be a sequence of vector spaces of $\setQ^d$, and let $\vec{x}_1,\ldots,\vec{x}_{k+1}$ be a sequence of vectors in $\setQ^d$ such that the following inclusion holds.
  $$\vec{P}\subseteq \bigcup_{j=1}^{k+1}\vec{x}_j+\vec{V}_j$$
  If $\vec{P}\subseteq \vec{x}_{k+1}+\vec{V}_{k+1}$ then we are done. So, we can assume that there exists $\vec{p}_0\in \vec{P}$ such that $\vec{p}_0\not\in \vec{x}_{k+1}+\vec{V}_{k+1}$. We introduce the set $J=\{j\mid \vec{p}_0\in \vec{x}_j+\vec{V}_j\}$. Notice that $|J|\leq k$ since $\vec{p}_0\not\in \vec{x}_{k+1}+\vec{V}_{k+1}$, and $J$ is non empty since $\vec{p}_0\in\vec{P}\subseteq\bigcup_{j=1}^{k+1}\vec{x}_j+\vec{P}_j$. Let $\vec{p}\in \vec{P}$ and $n\in\setN$. As $\vec{p}_0+n\vec{p}\in \vec{P}$ we deduce that there exists $j\in\{1,\ldots,k+1\}$ such that $\vec{p}_0+n\vec{p}\in \vec{x}_j+\vec{V}_j$. Since $\{1,\ldots,k+1\}$ is finite while $\setN$ is infinite, there exists $j\in\{1,\ldots,k+1\}$ and $n<m$ in $\setN$ such that $\vec{p}_0+n\vec{p}$ and $\vec{p}_0+m\vec{p}$ are both in $\vec{x}_j+\vec{V}_j$. The difference of those two vectors shows that $\vec{p}\in\vec{V}_j$ and in particular $\vec{p}_0\in \vec{x}_j+\vec{V}_j$. So $j\in J$. We have proved that $\vec{P}\subseteq \bigcup_{j\in J}\vec{V}_j$. By induction hypothesis, we deduce that there exists $j\in J$ such that $\vec{P}\subseteq \vec{V}_j$. It follows that $\vec{p}_0\in \vec{V}_j$. Moreover, from $j\in J$, we also get $\vec{p}_0\in \vec{x}_j+\vec{V}_j$. Hence $\vec{x}_j\in \vec{V}_j$. It follows that $\vec{V}_j=\vec{x}_j+\vec{V}_j$. From $\vec{P}\subseteq \vec{V}_j$ we derive $\vec{P}\subseteq \vec{x}_j+\vec{V}_j$. The induction is proved.
\end{proof}

\cordimspacedim*
\begin{proof}
  Let $\vec{V}$ be the vector space spanned by $\vec{P}$. From $\vec{P}\subseteq\vec{V}$ we deduce that $\sdim{\vec{P}}\leq \dim{\vec{V}}$ by minimality of $\sdim{\vec{P}}$. Now, let $\vec{V}_1,\ldots,\vec{V}_k$ be a sequence of vector spaces of dimension at most $\sdim{\vec{P}}$, and $\vec{x}_1,\ldots,\vec{x}_k$ be a sequence of vectors such that $\vec{P}\subseteq\bigcup_{j=1}^k\vec{x}_j+\vec{V}_j$. As $\vec{P}$ is non-empty, we deduce that $k\geq 1$. Lemma~\ref{lem:insecable} shows that there exists $j$ such that $\vec{P}\subseteq \vec{x}_j+\vec{V}_j$. Let $\vec{p}_0\in \vec{P}$. Notice that $\vec{p}_0\in \vec{x}_j+\vec{V}_j$. Since $\vec{p}_0+\vec{P}\subseteq\vec{P}\subseteq \vec{x}_j+\vec{V}_j$ we deduce that $\vec{P}\subseteq \vec{x}_j-\vec{p}_0+\vec{V}_j=\vec{V}_j$ since $\vec{p}_0\in\vec{x}_j+\vec{V}_j$. We deduce that $\vec{V}\subseteq \vec{V}_j$ by minimality of the vector space spanned by $\vec{P}$. Therefore $\dim{\vec{V}}\leq\dim{\vec{V}_j}$. The previous inequality with $\dim{\vec{V}_j}\leq\sdim{\vec{P}}$ provides $\dim{\vec{V}}\leq\sdim{\vec{P}}$. We have proved that $\sdim{\vec{P}}=\dim{\vec{V}}$.
\end{proof}

\lemspacedifflinear*
\begin{proof}
  Let $(\vec{b},\vec{G})$ be a linear-presentation of $\vec{L}$ and let $\vec{P}=\per{\vec{G}}$. As $\vec{L}+\vec{q}\subseteq \vec{L}$, we deduce that $\vec{b}+\vec{q}\in \vec{b}+\vec{P}$, i.e. $\vec{q}\in\vec{P}$. As $\vec{L}\setminus (\vec{L}+\vec{q})$ is equal to $\vec{b}+(\vec{P}\setminus (\vec{P}+\vec{q}))$, it is sufficient to prove that $\sdim{\vec{P}\setminus (\vec{P}+\vec{q})}<\sdim{\vec{P}}$. 
  
  Let $\vec{V}$ be the vector-space spanned by $\vec{P}$. Corollary~\ref{cor:dimspacedim} shows that $\sdim{\vec{P}}=\dim{\vec{V}}$. Let $\vec{g}_1,\ldots,\vec{g}_k$ be an enumeration of $\vec{G}$, i.e. such that $\vec{G}=\{\vec{g}_1,\ldots,\vec{g}_k\}$. We denote by $\vec{V}_J$ the vector space spanned by $(\vec{g}_j)_{j\in J}$ for every $J\subseteq \{1,\ldots,k\}$, and we let $\mathcal{F}$ be the set of $J\subseteq\{1,\ldots,k\}$ such that $\vec{V}_J=\vec{V}$. Notice that $\dim{\vec{V}_J}<\dim{\vec{V}}$ for every $J\not\in \mathcal{F}$, since $\vec{V}_J$ is strictly included in $\vec{V}$ in that case.

  Let $J\in\mathcal{F}$ and let us prove that there exists $n_J\in\setN$ such that $n_J\sum_{j\in J}\vec{g}_j\in \vec{P}+\vec{q}$. As $(\vec{g}_j)_{j\in J}$ span the vector space $\vec{V}$, there exists a sequence $(\lambda_j)_{j\in J}$ of rational numbers such that $\vec{q}=\sum_{j\in J}\lambda_j\vec{g}_j$. In particular, there exists $m\in\setN_{>0}$ such that $m\lambda_j\in\setZ$ for every $j\in J$. Let $n_J\in\setN$ such that $n_J\geq m\lambda_j$ for every $j\in J$. We have $n_J\sum_{j\in J}\vec{g}_j=\sum_{j\in J}m\lambda_j\vec{g}_j+\sum_{j\in J}(n_J-m\lambda_j)\vec{g}_j$. Hence $n_J\sum_{j\in J}\vec{g}_j\in m\vec{q}+\vec{P}\subseteq \vec{P}+\vec{q}$.

  Now, let $n=\max_{J\in \mathcal{F}}n_J$ and observe that $\sum_{j\in J}n_j\vec{g}_j\in \vec{q}+\vec{P}$ for every $J\in\mathcal{F}$ and for every sequence $(n_j)_{j\in J}$ such that $n_j\geq n$ for every $j\in J$. Let $\vec{p}\in \vec{P}\setminus (\vec{q}+\vec{P})$. There exists a sequence $n_1,\ldots,n_k\in\setN$ such that $\vec{p}=\sum_{j=1}^k n_j\vec{g}_j$. Let $J=\{j\mid n_j\geq n\}$. If $J\in\mathcal{F}$ then $\vec{p}\in \vec{q}+\vec{P}$ and we get a contradiction. Hence, $J\not\in \mathcal{F}$. We have proved the following inclusion:
  $$\vec{P}\setminus(\vec{P}+\vec{q})\subseteq \bigcup_{J\not\in\mathcal{F}}\sum_{j\not\in J}\{0,\ldots,n-1\}\vec{g}_j+\vec{V}_J$$
  In particular $\sdim{\vec{P}\setminus (\vec{P}+\vec{q})}\leq \max_{J\not\in \mathcal{F}}\dim{\vec{V}_J}$ where the max returns $-1$ it the set of indexes is empty. Since $\max_{J\not\in \mathcal{F}}\dim{\vec{V}_J}<\dim{\vec{V}}=\sdim{\vec{P}}$, we have proved the lemma.
\end{proof}

\section{Proofs of Section~\ref{sec:periovas}}
\lemundecper*
\begin{proof}
  The VAS reachability set inclusion problem takes as input two VAS $(\vec{A}_1,\{\vec{c}_1\})$ and $(\vec{A}_2,\{\vec{c}_2\})$ having the same dimension, i.e. such that $\vec{A}_1,\vec{A}_2$ are two finite sets of actions in $\setZ^d$, and $\vec{c}_1,\vec{c}_2\in\setN^d$, and checks whether $\poststar{\vec{A}_1}{\vec{c}_1}\subseteq \poststar{\vec{A}_2}{\vec{c}_2}$. This problem is known to be undecidable~\cite{Baker73,DBLP:journals/tcs/Hack76,DBLP:journals/tcs/Jancar95}. We are going to reduce the VAS reachability set inclusion problem to the problem of deciding if the projection $\vec{X}$ of the reachability set of a VAS $(\vec{A},\vec{C}_0)$ defined below is periodic. 

  This VAS will have 6 additional counters. To simplify the presentation, vectors in $\setZ^{d+6}$ are denoted as tuples in $\setZ^d\times\setZ^6$.

  The set of actions $\vec{A}$ is defined as follows.
  \begin{align*}
    \vec{A} =
    \vec{A}_1\times \{(0,0,-1,1,0,0), (0,0,1,-1,0,0)\}\cup
    \vec{A}_2\times \{(0,0,0,0,-1,1), (0,0,0,0,1,-1)\}
  \end{align*}
  The semilinear set $\vec{C}_0$ is defined as follows where $\vec{F}=\{(\vec{0},(0,0,0,0,0,0)),(\vec{0},(0,1,0,0,0,0))\}\cup \setN^d\times \vec{E}\times \setN^4$ and $\vec{E}=\{(m,n)\in\setN^2 \mid m\geq 2\vee m+n\geq 3\}$.
  \begin{align*}
    \vec{C}_0 = \vec{F}\cup 
    &\{(\vec{c}_1,(1,0,1,0,0,0)),(\vec{c}_1,(1,0,0,1,0,0))\}\cup\\
    &\{(\vec{c}_2,(1,1,0,0,1,0)), (\vec{c}_2,(1,1,0,0,0,1))\}
  \end{align*}
  Since $\vec{F}$ is an inductive invariant for $\vec{A}$, we easily derive the following equality.
  \begin{align*}
    \poststar{\vec{A}}{\vec{C}_0}=
    \vec{F}\cup
    & \poststar{\vec{A}_1}{\{\vec{c}_1\}}\times \{(1,0,1,0,0,0), (1,0,0,1,0,0)\}\cup \\
    & \poststar{\vec{A}_2}{\{\vec{c}_2\}}\times \{(1,1,0,0,1,0), (1,1,0,0,0,1)\}
  \end{align*}
  In particular, the set $\vec{X}\coloneqq\{(\vec{x},(m,n))\in\setN^d\times\setN^2 \mid \exists \vec{y}\in\setN^4,
  (\vec{x},m,n,\vec{y})\in \poststar{\vec{A}}{\vec{c}_0}\}$ is equal to the following set:
  \begin{align*}
    \{(\vec{0},(0,0)),(\vec{0},(0,1))\}\cup (\setN^d\times\vec{E})\cup
     (\poststar{\vec{A}_1}{\{\vec{c}_1\}}\times\{(1,0)\})\cup
     (\poststar{\vec{A}_2}{\{\vec{c}_2\}}\times\{(1,1)\})
  \end{align*}
  Observe that if $\vec{X}$ is periodic, then $(\poststar{\vec{A}_1}{\{\vec{c}_1\}}\times \{(1,0)\}) +(\vec{0},(0,1))$ is included in $\vec{X}$. In particular $\poststar{\vec{A}_1}{\{\vec{c}_1\}}\subseteq \poststar{\vec{A}_2}{\{\vec{c}_2\}}$ since $(1,1)\not\in\vec{E}$. Conversely, if the previous inclusion holds, notice that $\vec{X}$ is periodic.
\end{proof}

\section{Proofs About Interior Vectors}\label{app:intPQ}
\lemintcara*
\begin{proof}
    Let us denote by $\vec{V}$ the vector space spanned by $\vec{P}$. 
    
    Assume first that $\vec{p}$ is an interior vector of $\vec{P}$ and let $\vec{p}_1,\ldots,\vec{p}_d$ be a sequence of vectors in $\vec{P}$ spanning the vector space $\vec{V}$, and let $\vec{q}=\sum_i\vec{p}_i$. Since $\vec{p}$ is an interior vector, there exists $n\in\setN_{>0}$ such that $n\vec{p}\in \vec{q}+\vec{P}$. Let us introduce $\vec{p}_{d+1}\in\vec{P}$ such that $n\vec{p}=\vec{q}+\vec{p}_{d+1}$. Since $\vec{p}=\frac{1}{n}(\vec{p}_1+\cdots+\vec{p}_k)$ with $k=d+1$, we have proved one direction of the lemma.

    For the other direction, assume that $\vec{p}\in\setN^d$ is a vector such that $\vec{p}=\sum_j\lambda_j\vec{p}_j$ for a sequence $\vec{p}_1,\ldots,\vec{p}_k$ of vectors in $\vec{P}$ that spans the vector space $\vec{V}$, and a sequence $\lambda_1,\ldots,\lambda_k\in\setQ_{>0}$. Let us prove that $\vec{p}$ is an interior vector. By replacing $\vec{p}$ and the vectors $\vec{p}_1,\ldots,\vec{p}_k$ by some multiples, we can assume w.l.o.g that $\vec{p}=\sum_j\vec{p}_j$. Now let $\vec{q}\in\vec{P}$. Since $\vec{q}\in\vec{V}$ we have $m\vec{q}=\sum_{j}z_j\vec{p}_j$ for some $z_1,\ldots,z_k\in\setZ$ and $m\in\setN_{>0}$. Let $z\in\setN_{>0}$ such that $n_j:=z-z_j$ is in $\setN$ for every $j$. Observe that $z\vec{p}-m\vec{q}=\sum_jn_j\vec{p}_j\in\vec{P}$. It follows that $z\vec{p}\in \vec{q}+(m-1)\vec{q}+\vec{P}\subseteq \vec{q}+\vec{P}$. We have proved that $\vec{p}$ is an interior vector.
\end{proof}

\lemintPQ*
\begin{proof}
 We introduce $\vec{V}:=\setQ\vec{P}$ the vector space spanned by $\vec{P}$. Since $\setQ_{\geq 0}\vec{P}$ is included in the vector space $\vec{V}$, which is topologically closed (recall by duality that any vector space is the set of vectors satisfying a system of homogeneous linear equalities), it follows that $\overline{\setQ_{\geq 0}\vec{P}}$ is also included in $\vec{V}$ by minimality of the topological closure. Hence $\vec{Q}\subseteq \vec{V}$ and we deduce that $\vec{V}$ is also the vector space spanned by $\vec{Q}$. In particular, by Lemma~\ref{lem:intcara}, since $\vec{P}\subseteq \vec{Q}$, we deduce that any interior vector of $\vec{P}$ is an interior vector of $\vec{Q}$.

 Now, let $\vec{q}\in\setN^d$ be a vector interior to $\vec{Q}$. We consider a sequence $\vec{p}_1,\ldots,\vec{p}_k$ of vectors in $\vec{P}$ spanning the vector space $\vec{V}$, and we let $\vec{p}=\sum_j\vec{p}_j$. Let us denote by $B_{\epsilon}=\{\vec{x}\in\vec{V} \mid \norm{\vec{x}}<\epsilon\}$ the \emph{open ball} of $\vec{V}$ of radius $\epsilon\in\setQ_{>0}$ centered on zero (where $\norm{x}=\max_i|x(i)|$). Since $\vec{p}_1,\ldots,\vec{p}_k$ is spanning the vector space $\vec{V}$, Cramer's rules show that there exists $\epsilon\in\setQ_{>0}$ such that the following inclusion holds where $(-1,1)$ denotes the open interval $\{\lambda\in\setQ \mid -1<\lambda<1\}$:
    $$B_\epsilon\subseteq \sum_{j}(-1,1)\vec{p}_j$$

Since $\vec{q}$ is an interior vector of $\vec{Q}$, there exists $n\in\setN_{>0}$ and $\vec{x}\in\vec{Q}$ such that $n\vec{q}=\vec{p}+\vec{x}$. As $\vec{x}\in \overline{\setQ_{\geq 0}\vec{P}}$, there exists $\vec{y}\in \setQ_{\geq 0}\vec{P}$ such that $\vec{z}\coloneqq\vec{x}-\vec{y}$ is in $B_\epsilon$. Notice that $n\vec{q}=(\vec{p}+\vec{z})+\vec{y}$ and we have:
$$\vec{p}+\vec{z}\in \sum_{j}(0,2)\vec{p}_j$$
It follows from Lemma~\ref{lem:intcara} that $\vec{q}$ is an interior vector of $\vec{P}$.
\end{proof}

\lempushin*
\begin{proof}
  Lemma~\ref{lem:intcara} shows that $\vec{x}$ is in the interior of $\vec{Q}$ and Lemma~\ref{lem:intPQ} shows that $\vec{x}$ is in the interior of $\vec{P}$. In particular there exists $n\in\setN_{>0}$ such that $n\vec{x}\in \vec{0}+\vec{P}$. By replacing $\vec{x}$ by a multiple, we can assume w.l.o.g that $\vec{x}\in\vec{P}$. As $\vec{y}\in \vec{Q}\subseteq \vec{P}-\vec{P}$, there exists $\vec{p}\in\vec{P}$ such that $\vec{y}+\vec{p}\in\vec{P}$. As $\vec{x}$ is in the interior of $\vec{P}$, there exists $r\in\setN_{>0}$ such that $r\vec{x}\in \vec{p}+\vec{P}$. It follows that $\vec{y}+r\vec{x}\in \vec{y}+\vec{p}+\vec{P}\subseteq \vec{P}$. As $\vec{x}\in\vec{P}$, we deduce that $\vec{y}+r\setN_{>0}\vec{x}\subseteq \vec{P}$.
\end{proof}

\section{What About Branching VAS and Grammar VAS?}\label{sec:appliBGVAS}
The recent approach of Leroux \cite{DBLP:conf/popl/Leroux11,Turing-100:Vector_Addition_Systems_Reachability}
to solve the VAS reachability problem is based on the fundamental property
(recalled in \cref{thm:inv}) that
every semilinear set $\vec{S}$ containing the reachability set of a VAS $(\vec{A}, \vec{C}_0)$ also contains a semilinear inductive invariant $\vec{I}$ for $(\vec{A}, \vec{C}_0)$.
Let us call this property the ``semilinear inductive invariant'' property.
In terms of the constraint system presented in \cref{subsec:vas},
this property says that
every semilinear set that is $\subseteq$-larger than the $\subseteq$-least solution of \cref{eq:inductive:init,eq:inductive:step} is also $\subseteq$-larger than some semilinear solution $\vec{I}$.
In this paper,
we have shown that $\vec{I}$ can in addition be required to be periodic when the $\subseteq$-least solution is periodic
(see \cref{thm:invper}).
We now discuss in this section the hypothetical extension of the ``semilinear inductive invariant'' property to branching VAS and grammar VAS.

\subsection{Branching VAS}

An \emph{initialized $d$-dim binary branching VAS} (\emph{$d$-BVAS} or just \emph{BVAS} for short) is
a triple $(\vec{A}_1, \vec{A}_2, \vec{C}_0)$ where
$\vec{A}_1 \subseteq \setZ^d$ is a finite set of \emph{unary actions},
$\vec{A}_2 \subseteq \setZ^d$ is a finite set of \emph{binary actions} and
$\vec{C}_0\subseteq \setN^d$ is a semilinear set of \emph{initial configurations}.
The \emph{reachability set} of a BVAS $(\vec{A}_1, \vec{A}_2, \vec{C}_0)$ is the $\subseteq$-least solution of the following monovariate system of constraints:
\begin{align}
  & \vec{X} \supseteq \vec{C}_0 \label{eq:inductive:init:bvas} \\
  & \vec{X} \supseteq (\vec{X}+\vec{a}) \cap \setN^d & \forall \vec{a} \in \vec{A}_1 \label{eq:inductive:ustep:bvas} \\
  & \vec{X} \supseteq (\vec{X}+\vec{X}+\vec{a}) \cap \setN^d & \forall \vec{a} \in \vec{A}_2 \label{eq:inductive:bstep:bvas}
\end{align}
An \emph{inductive invariant} for $(\vec{A}_1, \vec{A}_2, \vec{C}_0)$ is a solution of this constraint system.
Observe that every inductive invariant is periodic when $\vec{A}_2$ and $\vec{C}_0$ both contain the zero vector.
We conjecture that BVAS also satisfy the ``semilinear inductive invariant'' property.

\begin{conjecture}
  Every semilinear set $\vec{S}$ containing the reachability set of a BVAS $(\vec{A}_1, \vec{A}_2, \vec{C}_0)$ also contains a semilinear inductive invariant for $(\vec{A}_1, \vec{A}_2, \vec{C}_0)$.
\end{conjecture}

To test this conjecture,
one might want to consider the subcase where $\vec{A}_2 = \{\vec{0}\}$ and $\vec{C}_0$ is periodic.
In that case,
the reachability set of the BVAS $(\vec{A}_1, \vec{A}_2, \vec{C}_0)$ coincides with
the reachability set of the periodic VAS $(\vec{A}_1, \vec{C}_0)$.
Our \cref{thm:invper} entails that the above conjecture holds for that subcase.

\subsection{GVAS}

A \emph{$d$-dim grammar-controlled VAS} (\emph{$d$-GVAS} or just \emph{GVAS} for short) is
a context-free grammar $(V, \vec{A}, R, S)$ where
$V$ is a finite set of nonterminals,
$\vec{A} \subseteq \setZ^d$ is a finite set of \emph{terminals},
$R \subseteq V \times (\vec{A} \cup VV)$ is a finite set of production \emph{rules}\footnote{We consider, w.l.o.g., context-free grammars in Chomsky normal form.}, and
$S \in V$ is the \emph{start symbol}.
The \emph{reachability relation} of a GVAS $(V, \vec{A}, R, S)$ is the $\subseteq$-least solution of the following multivariate (one variable per nonterminal) system of constraints:
\begin{align}
  & \vec{X} \supseteq (\vec{a}^{-}, \vec{a}^{+}) & \forall (X \rightarrow \vec{a}) \in R \label{eq:inductive:A:gvas}\\
  & \vec{X} \supseteq \vec{X} + \{(\vec{v}, \vec{v}) \mid \vec{v} \in \setN^d\} \label{eq:inductive:diag:gvas}\\
  & \vec{X} \supseteq \vec{Y} \fatsemi \vec{Z} & \forall (X \rightarrow YZ) \in R \label{eq:inductive:YZ:gvas}
\end{align}
The vectors $\vec{a}^{-}$ and $\vec{a}^{+}$ in \cref{eq:inductive:A:gvas} are defined by
$\vec{a}^{-}(i) = 0$ and $\vec{a}^{+}(i) = \vec{a}(i)$ if $\vec{a}(i) \geq 0$, and
$\vec{a}^{-}(i) = -\vec{a}(i)$ and $\vec{a}^{+}(i) = 0$ if $\vec{a}(i) \leq 0$.
The symbol $\fatsemi$ in \cref{eq:inductive:YZ:gvas} stands for the forward relational composition.
A \emph{inductive relational invariant} for $(V, \vec{A}, R, S)$ is a solution of this constraint system.
We conjecture that GVAS satisfy the ``semilinear inductive relational invariant'' property.
The partial order $\subseteq$ over subsets of $\setN^d \times \setN^d$ is extended to families $\left(\vec{X}\right)_{X \in V}$ of subsets of $\setN^d \times \setN^d$ component-wise.

\begin{conjecture}
  Every family $\left(\vec{X}\right)_{X \in V}$ of semilinear sets containing the reachability relation of a GVAS $(V, \vec{A}, R, S)$ also contains a semilinear inductive relational invariant for $(V, \vec{A}, R, S)$.
\end{conjecture}

To test this conjecture,
one might want to consider the subcase where
$V$ contains a single nonterminal (i.e., $V = \{S\}$),
$\vec{A}$ contains the zero vector, and
$R = V \times (\vec{A} \cup VV)$ (i.e., $R$ contains the rule $S \rightarrow SS$ and all the rules $S \rightarrow \vec{a}$ with $\vec{a} \in \vec{A}$).
In that case,
the reachability relation $\vec{U}$ of the GVAS $(V, \vec{A}, R, S)$ coincides with
the reachability relation of the VAS $\vec{A}$.
As mentioned in \cref{subsec:diagonal-relations},
the latter coincides with the reachability set of the periodic VAS $(\vec{A}',\{(\vec{0},\vec{0})\})$ where
$\vec{A}' = (\{\vec{0}\}\times\vec{A})\cup \{(\vec{e}_i,\vec{e}_i)\mid 1\leq i\leq d\}$.
Given a semilinear set $\vec{X}$ containing $\vec{U}$,
our \cref{thm:invper} entails that $(\vec{A}',\{(\vec{0},\vec{0})\})$ admits a periodic semilinear inductive invariant $\vec{I}$ contained in $\vec{X}$.
Note that $\vec{U} \subseteq \vec{I}$ and $\vec{I}$ is diagonal since it is periodic and reflexive.
However,
$\vec{I}$ is not necessarily an inductive relational invariant for $(V, \vec{A}, R, S)$
as $\vec{I}$ may not be transitive.
We conjecture that every semilinear set containing the reachability relation $\vec{U}$ of a given VAS
also contains a semilinear overapproximation of $\vec{U}$ that is diagonal and transitive (and, hence, periodic).

}{}

\end{document}